\documentclass[12pt]{amsart}
   	
\usepackage{graphicx}

\newtheorem{theorem}[equation]{Theorem}
\newtheorem{lemma}[equation]{Lemma}
\newtheorem{corollary}[equation]{Corollary}

\title[Optimal Retirement Tontines]{Optimal Retirement Tontines for the 21st Century: \\ With Reference to Mortality Derivatives in 1693}
\author[M.A. Milevsky and T.S. Salisbury]{Moshe A. Milevsky and Thomas S. Salisbury}
\thanks{Milevsky is an Associate Professor of Finance at the Schulich School of Business,
York University, and Executive Director of the IFID Centre. Salisbury is a Professor in the Department of Mathematics and Statistics at York University. The authors acknowledge funding from a Schulich Research Fellowship (Milevsky) and from NSERC (Salisbury), and wish to thank Rejo Peter, Dajena Collaku, Simon Dabrowski and Branislav Nikolic for research assistance. The contact author (Milevsky) can be reached at: milevsky@yorku.ca. } 
\date{28 May 2013 (Draft version 1.6)}				

\begin{document}
\maketitle

\begin{abstract} 

Historical tontines promised enormous rewards to the last survivors at the expense of those who died early. And, while this design \emph{appealed to the gambling instinct}, it is a suboptimal way to manage longevity risk during retirement. This is why fair life annuities making constant payments -- where the insurance company is exposed to the longevity risk -- induces greater lifetime utility. However, tontines do not have to be designed using a winner-take-all approach and insurance companies do not actually sell fair life annuities, partially due to aggregate longevity risk. 

In this paper we derive the tontine structure that maximizes lifetime utility, but doesn't expose the sponsor to any longevity risk. Technically speaking we solve the Euler Lagrange equation and examine its sensitivity to (i.) the size of the tontine pool, (ii.) individual longevity risk aversion, and (iii.) subjective health status. The optimal tontine varies with the individual's longevity risk aversion $\gamma$ and the number of participants $n$, which is problematic for product design. That said, we introduce a structure called a \emph{natural tontine} whose payout declines in exact proportion to the (expected) survival probabilities, which is near-optimal for all $\gamma$ and $n$. We compare the utility of optimal tontines to the utility of loaded life annuities under reasonable demographic and economic conditions and find that the life annuity's advantage over tontines, is minimal. 

We also use our framework to review and analyze the first-ever mortality-derivative issued by the British government, known as {\it King Williams's tontine of 1693}. Although it is widely acknowledged that mortality-derivatives were mis-priced in their early years, it is worth noting that both life annuities and tontines co-existed during that period. We shed light on the preferences and beliefs of those who invested in the tontines vs. the annuities and conclude by arguing that tontines should be re-introduced and allowed to co-exist with life annuities. Individuals would likely select a portfolio of tontines and annuities that suit their personal preferences for consumption and longevity risk, as they did over 320 years ago.

\end{abstract}

\newpage

\begin{quote}
\emph{``...Upon the same revenue more money can always be raised by tontines than by annuities for separate lives. An annuity, with a right of survivorship, is really worth more than an equal annuity for a separate life, and from the confidence which every man naturally has in his own good fortune, the principle upon which is founded the success of all lotteries, such an annuity generally sells for something more than it is worth. In countries where it is usual for government to raise money by granting annuities, \underline{tontines} are upon this account generally \underline{preferred to annuities} for separate lives..."}
\end{quote}

\begin{flushright}
{\bf Adam Smith, The Wealth of Nations, 1776} \\
\end{flushright}

\section{Introduction and Executive Summary}

As policymakers, academics and the public at large grow increasingly concerned about the cost of an aging society, we believe it is worthwhile to go back in time and examine the capital market instruments used to finance retirement in a period before social insurance, defined benefit (DB) pensions and annuity companies. Indeed, in the latter part of the 17th century and for almost two centuries afterwards, one of the most popular retirement investments in the developed world was not a stock, bond or a mutual fund -- although they were available. In fact, the preferred method used by many individuals to generate income in the senior years of the lifecycle was a so-called tontine scheme sponsored by government\footnote{Sources: Weir (1989), Poterba (2005) and McKeever (2009)}. Part annuity, part lottery and part hedge fund, the tontine -- which recently celebrated its 360th birthday -- offered a lifetime of income that increased as other members of the tontine pool died off and their money was distributed to survivors.  

The underlying and original tontine scheme is quite distinct from its public image as a lottery for centenarians in which the longest survivor wins all the money in a pool. In fact, the tontine annuity -- as it is sometimes called -- is more subtle and much more elegant. Imagine a group of 1000 soon-to-be retirees who band together and pool \$1,000 each to purchase a million-dollar U.S. Treasury bond paying 3\% coupons. The bond generates \$30,000 in interest yearly, which is split among the 1000 participants in the pool, for a 30,000 / 1,000 = guaranteed \$30 dividend per member. A custodian holds the big bond -- taking no risk and requiring no capital -- and charges a trivial fee to administer the annual dividends. So far this structure is the basis for all bond funds. Nothing new. But, in a tontine scheme the members agree that if-and-when they die, their guaranteed \$30 dividend is split amongst those who still happen to be alive. 

For example, if one decade later only 800 original investors are alive, the \$30,000 coupon is divided into 800, for a \$37.50 dividend each. Of this, \$30 is the guaranteed dividend and \$7.50 is \emph{other people's money} Then, if two decades later only 100 survive, the annual cash flow to survivors is \$300, which is a \$30 guaranteed dividend plus \$280. When only 30 remain, they each receive \$1,000 in dividends, which mind you, is a 100\% yield in that year alone. The extra payments -- above and beyond the guaranteed \$30 dividend --  are the mortality credits. In fact, under this scheme payments are expected to increase at the rate of mortality, which is a type of super-inflation hedge.

The tontine, of course, differs from a conventional life annuity. Although both offer income for life and pool longevity risk, the mechanics and therefore the cost to the investor are quite different.  The annuity promises predictable guaranteed lifetime payments, but this comes at a cost  -- and regulator-imposed capital requirements -- that inevitably makes its way to the annuitant. Indeed, the actuaries make very conservative assumptions regarding how long annuitants are likely to live, and then spend the next five decades worrying whether they got it right.  In contrast, the tontine custodian divides the variable $X$ (bond coupons received) by the variable $Y$ (participants alive) and sends out checks\footnote{In between the tontine and the annuity lies the \emph{participating annuity} which shares some longevity risk within a pool using a long-term smoothing scheme. Unfortunately, the actual formula for determining the smoothing mechanism is anything but smooth and has little basis in economics. That said, self annuitization schemes (GSA), proposed by Piggott, Valdez and Detzel (2005) have been growing in academic popularity.}. It's cleaner to administer, less capital-intensive and -- in its traditional form -- results in an increasing payment stream over time (assuming you are alive of course). 

In this paper we make the argument that properly designed tontines should be on the menu of products available to individuals as they transition into their retirement years.

\subsection{History vs. The Future}

Lorenzo Tonti was a colorful Italian banker who in the 1650s invented and promoted the scheme which shares his name. He described it as a mixture of lottery and insurance equally driven by old-age fear and economic greed.  They were first introduced in Holland, then very successfully in France\footnote{Source: Jennings and Trout (1982)} and a few decades later, in 1693, England's King William of Orange presided over the first government-issued tontine. The goal was to raise a million pounds to help finance his war against France. For a minimum payment of \pounds 100, participants were given the option of (a) investing in a tontine scheme paying guaranteed \pounds 10 dividends until the year 1700 and then \pounds 7 thereafter, or (b) a life annuity paying \pounds 14 for life, but with no survivorship benefits. This choice -- between two possible ways of financing retirement -- is quite fascinating and the intellectual impetus for this paper.

Over a thousand Englishmen (but very few woman) decided to invest in the tontine. The sums involved weren't trivial. The \pounds 100 entry fee would be worth at least \$100,000 today\footnote{Source: Lewin (2003)}. So, this was no impulse lottery purchase. Rather, some investors picked the tontine because they wanted the skewness -- that is potential for a very large payout if they survived -- while others wanted the more predictable annuity income option. The oldest survivor of King William's tontine of 1693 lived to age 100, earning \pounds 1,000 in her final year -- and no doubt well tended to by her family\footnote{Source: Finlaison (1829)}. More on this intriguing episode, later.

Some might view this entire exercise as a mere historical and intellectual curiosity. After all, tontines are (effectively) illegal in the U.S. and certainly have a nasty taint associated with them. However, we believe that there are some deep and subtle lessons that one can learn about the design of retirement income products in the 21st century, from this episode in financial history. Indeed, the choice between a tontine in which longevity risk is pooled and a life annuity in which payments are guaranteed is a euphemism for many of the choices retirees, and society, now face. 

Once again, our main practical points will be that (i) tontines and life annuities have co-existed in the past, so perhaps they can in the future; (ii) there really is no reason to construct the tontine payout function such that the last survivor receives hundreds of multiples times their initial investment; and (iii) a properly constructed tontine can result in lifetime utility that is comparable to the utility of a life annuity. 

\subsection{Outline of the Paper}

The remainder of this paper is organized as follows. The next section (\ref{1693}) describes the first English tontine of 1693 in great detail, which is one of the earliest opportunities to examine the real-life choice between tontines and life annuities. Section (\ref{theory}) is the theoretical core of the paper, which derives the properties of the optimal tontine structure and compares it with the life annuity.  Section (\ref{numerics}) uses and applies the results from section (\ref{theory}) and addresses how 21st century tontines might be constructed. In section (\ref{lit}) we briefly survey the existing literature which wasn't directly referenced in earlier sections of the paper. Finally, section (\ref{conc}) concludes the paper and section (\ref{apen}) is an appendix which contains all non-essential proofs and derivations.

\section{King William's Tontine of 1693}
\label{1693}

On November 4th, 1692, during the fourth year of the reign of King William and Queen Mary, the British Parliament passed the so-called \emph{Million Act}, which was a rather desperate attempt to raise one million pounds towards carrying on the war against France.  The Million Act specified that any British native or foreigner -- except for French citizens, one might suppose -- could purchase a tontine share from the exchequer for \pounds 100, prior to May 1st, 1693 and thus gain entry into the first British government tontine scheme.

For \pounds 100 an investor could select any nominee of any age -- including the investor himself~-- on whose life the tontine would be contingent. Dividend payments would be distributed to the investor as long as the nominee was still alive. Now, to put the magnitude of the minimal \pounds 100 investment in perspective, the average annual wage of building laborers in England during the latter part of the 17th century was approximately \pounds 16 and a few shillings per year\footnote{Source: Lewin (2003)}. So the entry investment in the tontine pool far exceeded the average industrial wage and the annual dividends alone might serve as a decent pension for a common laborer. It is therefore quite plausible to argue that the 1693 tontine was an investment for the rich and perhaps even one of the first exclusive hedge funds.

This was a simpler structure compared to the original tontine scheme envisioned by Lorenzo de Tonti in the year 1653, which involved multiple classes and different dividend rates. In the 1693 English tontine, each share of \pounds 100 would entitle the investor to an annual dividend of \pounds 10 for seven years (until June 1700), after which the dividends would be reduced to \pounds 7 per share. The 10\% and 7\% tontine dividend rate exceeded prevailing (risk free) interest rates in England at the end of the 17th century, which were officially capped at 6\%\footnote{Source: Homer and Sylla (2005)}. Note the declining structure of the interest payments, which is a preview of our soon-to-come discussion about the optimal tontine payout function.

Interestingly, nowhere within the Statutes of the Realm, which reproduces the Million Act verbatim, was the word tontine actually mentioned and neither was Lorenzo de Tonti ever referenced. Rather, the act stated innocuously: \emph{And so, from time to time upon the death of every nominee, whatsoever share of dividend was payable during the life of such nominee shall be equally divided amongst the rest of the contributors}\footnote{British History Online, \emph{An Act for Granting to Their Majesties certain Rates and Duties of Excise upon Beer, Ale and other Liquors, William and Mary 1692. www.british-history.ac.uk}}.

Although the monarch's plan was to use the money to fund an expensive war, the same act also introduced a new excise tax on beer, ale and other liquors for a period of 99 years. The rates included nine pence for every brewed barrel of ale, six pence for every gallon of brandy and three pence for every hogshead of cider. The intent was to use this tax to cover the dividends on the tontine scheme. This is a very early example of securitization, or an attempt to borrow against future tax revenues.

Moving on to the annuity side of the offering, to further entice investors to participate in the tontine scheme, the act included a unique sweetener or bonus provision. It stipulated that if the entire \pounds 1,000,000 target wasn't subscribed by May 1693 -- thus reducing the size of the lottery payoff for the final survivor -- the investors who had enrolled in the tontine during the six month subscription period (starting in November 1692) would have the option of converting their \pounds 100 tontine shares into a life annuity paying \pounds 14 per year. Under this alternative, the 14\% dividend payments were structured as a conventional life annuity with no group survivorship benefits or tontine features. Think of a single premium life annuity. This option-to-convert is quite intriguing and worthy of further analysis.

Why this extra option was added to the Million Act is unclear. Perhaps the fear amongst the tontine administrators -- and reluctance expressed by potential investors -- was that otherwise few people would subscribe. Indeed, part of what entices speculators to join a tontine is the large potential payout, which accrues to the last survivor.  One requires more participants for a bigger jackpot, but participants aren't likely to join unless the jackpot is big enough to begin with. So, the option-to-convert was likely added to give investors a possible exit strategy if they were disappointed with the number of subscribers. Naturally, an investor can't just cash-out of a tontine pool and walk away with original capital once it is up and running. 

Alas, the funds raised by early May 1963 fell far short of the million pound target. According to records maintained by the Office of the Exchequer, stored within the Archives of the British Library\footnote{Source: Howard, (1694)}, a total of only \pounds 377,000 was subscribed and approximately 3,750 people were nominated\footnote{Note that one nominee could have multiple shares of \pounds 100 contingent on their life.} to the tontine during the period between November 1692 and May 1693. This then triggered the option-to-exchange the tontine into a 14\% life annuity, because, obviously, the \pounds 1,000,000 target wasn't reached.

Interestingly, it seems that a total of 1,013 nominees (representing 1081 tontine shares) remained in the original 10\%/7\% tontine, while the other third elected to convert their tontine shares into a 14\% life annuity contingent on the same nominee. The following table \ref{table01} and table \ref{table02} provide a summary of the number of nominees who stayed in the tontine as well as the number of shares they purchased.

\begin{center}
{\bf Table \ref{table01} and Table \ref{table02} here}
\end{center}
 
In fact, the government passed another act in June 1693, to make the 14\% life annuities available to anyone in a desperate attempt to reach their funding target. They still couldn't raise a million pounds. Remember, interest rates at the time were 6\%, and there were absolutely no restrictions placed on the nominee's age for either the tontine or the annuity.

At first glance, it is rather puzzling why anyone would stay in the tontine pool instead of switching to the life annuity. On a present value basis, a cash flow of \pounds10 for 7 years and \pounds 7 thereafter, is much less valuable compared to \pounds 14 for life. As we shall describe later, the actuarial present value of the 10\%/7\% combination at the 6\% official interest rate was worth approximately \pounds 133 at the (typical nominee) age of ten, whereas the value of the life annuity was worth almost \pounds 185. Remember, the original investment was \pounds 100. It is no surprise that the British government was really losing money on these 14\% annuities\footnote{Source: Finlaison (1829)}. And, they were offering these terms to anyone regardless of how young they were. Why did anyone stay in the tontine pool?

It's not just 21st century financial logic that dictates that a 14\% life annuity would have been a good (and better) deal. In fact, none other than the astronomer Edmond Halley, writing in the January 1693 edition of the \emph{Philosophical Transactions of the Royal Society} just a few months after the passage of the Million Act opined on the matter. He wrote that:

\begin{quote}
\emph{This shows the great advantage of putting money into the present fund granted to their majesties, giving 14\% per annum, of at the rate of 7 years purchase for a life, when [even] the young lives at the annual [6\%] rate of interest, are worth above 13 years purchase.}
\end{quote}

The phrase \emph{years purchase} is an early actuarial term for the number of years before one gets one's entire money back. It's another way of quoting a price. For example, if you receive \pounds 5 per year for life, it takes 20 years to get your \pounds 100 back (ignoring interest), so this would be called 20 years purchase.

A major research question then, is: Why did a full third of the subscribers decide to stay in the tontine pool, while two thirds switched to the life annuity? It's hard to argue that the remaining tontine investors were irrational, ignorant or perhaps taking the lottery ticket approach, while the life annuity investors took the better financial terms. After all, the \pounds 100 was a very large sum of money in 1693. Could differing views on mortality and/or risk aversion explain why some investors switched, while others didn't? Can one provide a rational explanation for this seemingly irrational choice? The framework we introduced in section (\ref{theory}) might help shed light on this decision, or at least rationalize the choice.

Note that although very careful records must have been kept for the purpose of administrating the tontine pool, today we only have access to (i.) a list of every single person who participated in the 1693 tontine, as well as the age of their nominees, (ii.) their status in the year 1730, and (iii.) their status in 1749. These are the three primary documents at our disposal. 

As mentioned in the introduction, the longest living nominee of the 1693 tontine -- who was a mere 10 years of age at the time of initial nomination -- lived to the age of 100, surviving for 90 more years to the year 1783.  She was a female who spent her senior years in Wimbledon and earned a dividend of \pounds 1,081 in her last year of life\footnote{Source: Finlaison (1829).}. Although her payout was capped once seven survivors remained in the tontine pool, her final payment was ten times the original investment of \pounds 100. Remember, this was only one year's worth of dividends! Note that had she switched over to the life annuity -- likely it was her father who made the decisions and nominated her, since she was only 10 years old at the time -- back in 1693 and then lived to 1783, her dividend would have been a mere \pounds 14. 

\begin{center}
{\bf Figure \ref{fig1} and \ref {fig1b} here}
\end{center}

The six month period in between the announcement of the \emph{Million Act} tontine scheme in November 1692 and the date at which the final list was closed in late May and early June 1693 was quite busy and interesting. The London-based promoters of the tontine -- perhaps expecting some sort of commission -- published a table in late 1692 purporting to show a low expected number of survivors and a correspondingly high dividend payout rate over the next 100 years. These initial projections were probably viewed as unattractive because the subscription rate was much, much lower than the 10,000 target. A few months later, in early 1693, the promoters published a follow-up table showing an even lower mortality rate for the group and a correspondingly higher projected dividend payout rate for survivors\footnote{Source: Lewin (2003)}. These two sets of tables, together with the actual experience of the 1013 nominees who remained in the tontine, are displayed graphically in Figure \ref{fig1} and \ref{fig1b}. Note how the actual survival rate far exceeded the initial projections, which might have been based on population morality vs. the much healthier (anti-selected) nominee group. There is also some speculation that Edmond Halley, who coincidently presented the first known life table and annuity pricing model to the Royal Society in March 1693, was involved in the creation of the mortality tables underlying the investment projections. In fact, Edmond Halley and his life table are repeatedly referred to by participants in the tontine who subsequently complained about the low mortality rates and even made accusations of fraud. But, other than the coincidental timing, there is no evidence Halley participated in the tontine scheme or it's promotion\footnote{Source: Walford (1871)}.

\begin{center}
{\bf Table \ref{table03} here}
\end{center}

Table \ref{table03} displays the distribution of the ages of the nominees when the list closed in the summer of 1693. Recall that these are the 1013 nominees who did not switch over to the annuity paying 14\%. They remained in the tontine pool earning 10\% for seven years and 7\% thereafter. Notice the age distribution in 1693 and the fact that so many nominees were in their 20s, 30s, 40s and even one in their 50s. Table \ref{table03} also displays the survival status of this group in 1730 and 1749, which are the only two dates at which detailed documentation and a list of the nominees is available\footnote{Source: Leeson (1968)}.

In the next section we present an economic theory to describe and understand who might elect to participate in a tontine and who might choose a life annuity, as well as the properties of a tontine that are likely to generate the highest lifetime utility. To pre-empt the result, we will show that although the life annuity clearly dominated the tontine in terms of lifetime utility, a declining tontine payout structure is in fact optimal. Furthermore, it is possible that an investor who believes their nominee is much healthier than the remainder of the tontine pool may favour the tontine over the annuity.

\section{Tontine vs. Annuity: Economic Theory}
\label{theory}

We assume an objective survival function ${}_tp_x$, for an individual aged $x$ to survive $t$ years. The implications of subjective vs. objective survival rates will be addressed later. We assume that the tontine pays out continuously, as opposed to quarterly or monthly, although this doesn't really change the economics of the matter. The basic annuity involves annuitants (who are also the nominees) each paying \$1 to the insurer initially, and receiving in return an income stream of $c(t)\,dt$ for life. The constraint on annuities is that they are fairly priced, in other words that with a sufficiently large client base, the initial payments invested at the risk-free rate will fund the called-for payments in perpetuity. Later we discuss the implications of insurance loadings. Either way, there is a constraint on the annuity payout function $c(t)$, namely that 
\begin{equation}
\label{annuityconstraint}
\int_0^\infty e^{-rt}{}_tp_x\, c(t)\,dt=1.
\end{equation}
Though $c(t)$ is the payout rate per survivor, the payout rate per initial dollar invested is ${}_tp_x\,c(t)$. We will return to this later. 

Letting $u(c)$ denote the instantaneous utility of consumption (a.k.a. felicity function), a rational annuitant (with lifetime $\zeta$) having no bequest motives will choose a life annuity payout function for which $c(t)$ maximizes the discounted lifetime utility:
\begin{equation}
E[\int_0^\zeta e^{-rt}u(c(t))\,dt]=\int_0^\infty e^{-rt}{}_tp_x\, u(c(t))\,dt
\label{annuityutilityspecification}
\end{equation}
where $r$ is (also) the subjective discount rate (SDR), all subject to the constraint \eqref{annuityconstraint}. 

By the Euler-Lagrange theorem\footnote{Source: Elsgolc (2007) page \#51 or Gelfand and Fomin (2000), page \#15}, this implies the existence of a constant $\lambda$ such that 
\begin{equation}
e^{-rt}{}_tp_x\,u'(c(t))=\lambda e^{-rt}{}_tp_x \quad\text{for every $t$.}
\label{lagrangecondition}
\end{equation}
In other words, $u'(c(t))=\lambda$ is constant, so provided that utility function $u(c)$ is strictly concave, the optimal annuity payout function $c(t)$ is also constant. That constant is now determined by \eqref{annuityconstraint}, showing the following:
\begin{theorem} Optimized life annuities have constant $c(t)\equiv c_0$, where 
$$
c_0=\Big[\int_0^\infty e^{-rt}{}_tp_x\,dt\Big]^{-1}.
$$ 
\end{theorem} 
This result can be traced-back to Yaari (1965) who showed that the optimal (retirement) consumption profile is constant (flat) and that 100\% of wealth is annualized when there is no bequest motive. For more details and an alternate proof, see the excellent book by Cannon and Tonks (2008) and specifically the discussion on annuity demand theory in chapter 7.

\subsection{Optimal Tontine Payout}
In practice of course, insurance companies who are paying the life annuity $c(t)$ are exposed to both longevity risk, which is the uncertainty in ${}_tp_x$, as well as re-investment or interest rate risk, which is the uncertainty in $r$ over long horizons. The former is our concern here, so we will continue to assume that $r$ is a given constant for most of what follows. Note that even if re-investment rates were known with certainty, the insurance company would likely payout less than the $c(t)$ implied by equation \eqref{annuityconstraint} as a result of required capital and reserves, effectively lowering the lifetime utility of the (annuity and the) retiree.

This brings us to the tontine structures we will consider as an alternative, in which a predetermined dollar amount is shared among survivors at every $t$. Let $d(t)$ be the rate funds are paid out per initial dollar invested, a.k.a. the tontine payout function. Our main point (in this paper) is that there is no reason for tontine payout function to be a constant fixed percentage of the initial dollar invested (e.g. 4\% or 7\%), as it was historically. In fact, we can pose the same question as considered above for annuities: what $d(t)$ is optimal for subscribers, subject to the constraint that sponsor of the tontine cannot sustain a loss? Note that the natural comparison is now between $d(t)$ and ${}_tp_x\,c(t)$, where $c(t)$ is the optimal annuity payout found above. 

Suppose there are initially $n$ subscribers to the tontine scheme, each depositing a dollar with the tontine sponsor. Let $N(t)$ be the random number of live subscribers at time $t$. Consider one of these subscribers. Given that this individual is alive, $N(t)-1\sim \text{Bin}(n-1,{}_tp_x)$. In other words, the number of other (alive) subscribers at anytime $t$ is Binomially distributed with probability parameter ${}_tp_x$. 

So, as we found for the life annuity, this individual's discounted lifetime utility is
\begin{align*}
&E[\int_0^\zeta e^{-rt }u\Big(\frac{n d(t)}{N(t)}\Big)\,dt]=\int_0^\infty e^{-rt}{}_tp_x \,E[u\Big(\frac{n d(t)}{N(t)}\Big)\mid \zeta>t]\,dt\\
&\qquad=\int_0^\infty e^{-rt}{}_tp_x\sum_{k=0}^{n-1} \binom{n-1}{k}{}_tp_x^k(1-{}_tp_x)^{n-1-k}u\Big(\frac{nd(t)}{k+1}\Big)\,dt.
\end{align*}
The constraint on the tontine payout function $d(t)$ is that the initial deposit of $n$ should be sufficient to sustain withdrawals in perpetuity. Of course, at some point all subscribers will have died, so in fact the tontine sponsor will eventually be able to cease making payments, leaving a small windfall left over. But this time is not predetermined, so we treat that profit as an unavoidable feature of the tontine. Remember that we do not want to expose the sponsor to any longevity risk. It is the pool that bears this risk entirely.

Our budget or pricing constraint is therefore that 
\begin{equation}
\label{tontineconstraint}
\int_0^\infty e^{-rt} d(t)\,dt=1.
\end{equation}

So, for example, if $d(t)=d_0$ is forced to be constant (the historical structure, which we call a \emph{flat tontine}), then the tontine payout function (rate) is simply $d_0=r$, or slightly more if the upper bound of integration in \eqref{tontineconstraint} is less than infinity. We are instead searching for the optimal $d(t)$ which is far from constant.

By the Euler-Lagrange theorem from the Calculus of Variations, there is a constant $\lambda$ such that the optimal $d(t)$ satisfies 
\begin{equation}
e^{-rt}{}_tp_x\sum_{k=0}^{n-1} \binom{n-1}{k}{}_tp_x^k(1-{}_tp_x)^{n-1-k}\frac{n}{k+1}u'\Big(\frac{nd(t)}{k+1}\Big)=\lambda e^{-rt}
\label{opton1}
\end{equation}
for every $t$. Note that this expression directly links individual utility $u(c)$ to the optimal participating annuity. Recall that a tontine is an extreme case of participation or pooling of all longevity risk. Equation \eqref{opton1} dictates exactly how a risk averse retiree will tradeoff consumption against longevity risk. In other words, we are not advocating some ad-hoc actuarial process for smoothing realized mortality experience. 

Note that an actual mortality hazard rate $\mu_x$ does not appear in the above equation -- it appears only implicitly, in both ${}_tp_x$ and $\lambda$ (which is determined by \eqref{tontineconstraint}). Therefore, we will simplify our notation by re-parametrizing in terms of the probability: Let $D_u(p)$ satisfy 
\begin{equation}
p\sum_{k=0}^{n-1} \binom{n-1}{k}p^k(1-p)^{n-1-k}\frac{n}{k+1}u'\Big(\frac{nD_u(p)}{k+1}\Big)=\lambda.
\end{equation}
Substituting $p=1$ into the above equation, it collapses to $u'(D_u(1))=\lambda$. 

\begin{theorem}
Optimal tontine structure is $d(t)=D_u({}_tp_x)$, where $\lambda$ is chosen so \eqref{tontineconstraint} holds. 
\end{theorem}

We can simplify this in the case of Constant Relative Risk Aversion (CRRA) utility. Let  $u(c)=c^{1-\gamma}/(1-\gamma)$ if $\gamma\neq 1$, and when $\gamma= 1$ take $u(c)=\log c$ instead. Define 
\begin{equation}
\theta_{n,\gamma}(p)=E\Big[\Big(\frac{n}{N(p)}\Big)^{1-\gamma}\Big]=\sum_{k=0}^{n-1} \binom{n-1}{k}p^{k}(1-p)^{n-1-k}\Big(\frac{n}{k+1}\Big)^{1-\gamma}
\end{equation}
where $N(p)-1\sim\text{Bin}(n-1,p)$. Set $\beta_{n,\gamma}(p)=p\theta_{n,\gamma}(p)$. Then 
\begin{corollary} With CRRA utility, the optimal tontine has withdrawal rate
$D_{n,\gamma}^{\text{\rm OT}}(p)=D_{n,\gamma}^{\text{\rm OT}}(1)\beta_{n,\gamma}(p)^{1/\gamma}$, where
\begin{equation}
D_{n,\gamma}^{\text{\rm OT}}(1)=\Big[\int_0^\infty e^{-rt}\beta_{n,\gamma}({}_tp_x)^{1/\gamma}\,dt\Big]^{-1}.
\label{D(1)formula}
\end{equation}
\end{corollary}
\begin{proof}
Suppose $\gamma\neq 1$. Then the equation for $D_{n,\gamma}^{\text{OT}}(p)$  becomes that
\begin{equation}
D_{n,\gamma}^{\text{OT}}(p)^{-\gamma}p\theta_{n,\gamma}(p)=\lambda=D_{n,\gamma}^{OT}(1)^{-\gamma}.
\label{Dformula1}
\end{equation}
The constraint \eqref{tontineconstraint} now implies \eqref{D(1)formula}. A similar argument applies when $\gamma=1$. 
\end{proof}
$D_{n,\gamma}^{\text{OT}}(p)/D_{n,\gamma}^{\text{OT}}(1)=\beta_{n,\gamma}(p)^{1/\gamma}$ does not depend on the particular form of the mortality hazard rate $\mu_x$, or on the interest rate $r$, but only on longevity risk aversion $\gamma$ and the number of initial subscribers to the tontine pool, $n$. In other words, the mortality hazard rate and $r$ enter into the expression for $D_{n,\gamma}^{\text{OT}}(p)$ only via the constant $D_{n,\gamma}^{\text{OT}}(1)$.

To our knowledge, we are the first to derive the properties of an optimal tontine payout function, in contrast to the optimal life annuity which is well-known in the literature. Let us now compare the two. In other words, we compare the payout per initial subscriber for both products. Equivalently, we can compare the actual payout of the optimal tontine to what we call a {\it natural tontine}, in which the payout is $d(t)=D_{\text{N}}({}_tp_x)$ where $D_{\text{N}}(p)$ is proportional to $p$, just as is the case for the annuity payment per initial dollar invested. Comparing the budget constraints \eqref{annuityconstraint} and \eqref{tontineconstraint}, we see that  $D_{\text{N}}(p)=pc_0$. As we will soon see, the natural tontine is optimal for logarithmic utility ($\gamma=1$) and is close to optimal when $n$ is large. So we propose this as a reasonable structure for designing tontine products in practice. 

Figure \ref{fig6} shows the ratio of $D_{n,\gamma}^{\text{OT}}(p)$ to $D_{\text{N}}(p)$, with Gompertz hazard rate $\lambda(t)=\frac{1}{b}e^{\frac{x+t-m}{b}}$. This can be carried out for any specific hazard rate, but to get a universal result, we again re-parametrize in terms of $p$, and scale out constants. In other words, we consider the ratio 
\begin{equation}
R_{n,\gamma}(p)=\frac{D_{n,\gamma}^{\text{OT}}(p)/D_{n,\gamma}^{\text{OT}}(1)}{D_{\text{N}}(p)/D_{\text{N}}(1)}
=\big[p^{1-\gamma}\theta_{n,\gamma}(p)\big]^{1/\gamma}=\frac{\beta_{n,\gamma}(p)^{1/\gamma}}{p}.
\end{equation}
This has $R_{n,\gamma}(1)=1$ and depends only on $p$, $\gamma$, and $n$, not the specific mortality model. 
By the law of large numbers, it is clear that $\lim_{n\to\infty}R_{n,\gamma}(p)= 1$ for any fixed $\gamma$ and $p$. In general, we have the following:

\begin{theorem} For any $n$ and $0<p<1$, 
\label{Rtheorem}
$$
R_{n,\gamma}(p)\text{ is}
\begin{cases}
<1, & 0<\gamma<1\\
=1, & \gamma=1\\
>1, & 1<\gamma.
\end{cases}
$$
\end{theorem}
Since $R_{n,\gamma}(p)=\Big(\beta_{n,\gamma}(p)/p^\gamma\Big)^{1/\gamma}$, the theorem follows immediately from the following estimate,  which is proved in the appendix \ref{proof1}.
\begin{lemma} For any $n$ and $0<p<1$, 
\label{betabound}
$$
\beta_{n,\gamma}(p)\text{ is} 
\begin{cases}
<p^{\gamma}, &0<\gamma<1\\
=p^{\gamma}, & \gamma=1\\
>p^{\gamma}, & 1<\gamma.
\end{cases}
$$
\end{lemma}

Since low $p$ corresponds to advanced age, and $p=1$ corresponds to the date of purchase, Theorem \ref{Rtheorem} implies that individuals who are more longevity risk-averse than logarithmic (ie. $\gamma>1$) prefer to enhance the natural tontine payout at advanced ages, at the expense of the initial payout. Whereas individuals less longevity risk-averse than logarithmic (ie $0<\gamma<1$) prefer to enhance initial payouts at the expense of those at advanced ages. 

Observe that the case $\gamma=1$ of the Theorem implies the assertion made earlier, that the natural tontine is optimal when $\gamma=1$, The argument is simple in this case. For $u(c)=\log(c)$, the Euler-Lagrange equations are simply that $d(t)=\lambda\cdot{}_tp_x$, which implies that $D_{n,1}^{\text{OT}}(p)$ is proportional to $p$. Therefore $D_{n,1}^{\text{OT}}(p)=D_{\text{N}}(p)$ and so $R_{n,1}(p)=1$ for every $p$. 

In section \eqref{numerics} we will provide a variety of numerical examples that illustrate the optimal tontine payout function $d(t)$ as a function of longevity risk aversion 
$\gamma$ and the initial size of the tontine pool $n$.

\subsection{Optimal Tontine Utility vs. Annuity Utility}
Let $U_{n,\gamma}^{\text{OT}}$ denote the utility of the optimal tontine. To compute this, suppose $\gamma\neq 1$, and observe that
\begin{equation}
\frac{D_{n,\gamma}^{\text{OT}}(p)^{1-\gamma}}{1-\gamma}p\theta_{n,\gamma}(p)=\frac{D_{n,\gamma}^{\text{OT}}(p)}{1-\gamma}D_{n,\gamma}^{\text{OT}}(p)^{-\gamma}p\theta_{n,\gamma}(p)
=\frac{D_{n,\gamma}^{\text{OT}}(p)}{1-\gamma}D_{n,\gamma}^{\text{OT}}(1)^{-\gamma}
\end{equation}
by \eqref{Dformula1}. The utility of the optimal tontine is therefore precisely
\begin{align*}
U_{n,\gamma}^{\text{OT}}
&=\int_0^\infty e^{-rt}\frac{D_{n,\gamma}^{\text{OT}}({}_tp_x)^{1-\gamma}}{1-\gamma}{}_tp_x\,\theta_{n,\gamma}({}_tp_x)\,dt 
= \frac{D_{n,\gamma}^{\text{OT}}(1)^{-\gamma}}{1-\gamma}\int_0^\infty e^{-rt}D_{n,\gamma}^{\text{OT}}({}_tp_x)\,dt\\
&=\frac{D_{n,\gamma}^{\text{OT}}(1)^{-\gamma}}{1-\gamma}
=\frac{1}{1-\gamma}\Big(\int_0^\infty e^{-rt}\beta_{n,\gamma}({}_tp_x)^{1/\gamma}\,dt\Big)^\gamma
\end{align*}
by \eqref{tontineconstraint} and \eqref{D(1)formula}. 

Consider instead the utility $U_\gamma^{\text{A}}$ provided by the annuity, namely 
\begin{equation}
U_\gamma^{\text{A}}=\int_0^\infty e^{-rt}{}_tp_x \frac{c_0^{1-\gamma}}{1-\gamma}\,dt
=\frac{\int_0^\infty e^{-rt}{}_tp_x\,dt}{(1-\gamma)\Big(\int_0^\infty e^{-rt}{}_tp_x\,dt\Big)^{1-\gamma}}
=\frac{1}{1-\gamma}\Big(\int_0^\infty e^{-rt}{}_tp_x\,dt\Big)^\gamma.
\end{equation}

\begin{theorem}
$U_{n,\gamma}^{\text{\rm OT}}<U_\gamma^{\text{\rm A}}$ for any $n$ and $\gamma>0$.
\label{utilityinequality}
\end{theorem}
\begin{proof}
For $\gamma\neq 1$ this follows from Lemma \ref{betabound} and the calculations given above. We show the case $\gamma=1$ in the appendix. \end{proof}

\subsection{Indifference Annuity Loading}
The insurer offering an annuity will be modelled as setting aside some fraction of the initial deposits to fund the annuity's costs.  In other words, a fraction $\delta$ of the initial deposits are deducted initially, to fund risk management, capital reserves, etc. The balance, invested at the risk free rate $r$ will fund the annuity. Therefore, with loading, \eqref{annuityconstraint} becomes that $\int_0^\infty e^{-rt}{}_tp_x\, c(t)\,dt=1-\delta$, which implies that $c(t)\equiv c_1=(1-\delta)c_0$ is the optimal payout structure for the annuity. The utility of the loaded annuity is therefore 
$$
U^{\text{LA}}_{\gamma,\delta}=
\int_0^\infty e^{-rt}{}_tp_x \frac{c_1^{1-\gamma}}{1-\gamma}\,dt
=\frac{(1-\delta)^{1-\gamma}\int_0^\infty e^{-rt}{}_tp_x\,dt}{(1-\gamma)\Big(\int_0^\infty e^{-rt}{}_tp_x\,dt\Big)^{1-\gamma}}
=\frac{c_0^{-\gamma}}{1-\gamma}(1-\delta)^{1-\gamma}
$$
for $\gamma\neq 1$, and $\frac{\log(c_0)+\log(1-\delta)}{c_0}$ for $\gamma=1$. 

In section \ref{numerics} we will consider the \emph{indifference loading} $\delta$ that, when applied to the annuity, makes an individual indifferent between the annuity and a tontine, ie 
$U^{\text{LA}}_{\gamma,\delta}=U_{n,\gamma}^{\text{OT}}$. 
It turns out that the loading $\delta$ decreases with $n$, in such a way that the \emph{total} loading $n\delta$ stays roughly stable. In other words, there is at most a fixed amount (roughly) the insurer can deduct from the \emph{aggregate} annuity pool, regardless of the number of participants, before individuals start to derive greater utility from the tontine. We will illustrate this observation, at least for $1<\gamma\le 2$, by proving the following inequality in the appendix
\begin{theorem}
\label{loadinginequality}
Suppose that $1<\gamma\le 2$. Then 
$\delta<\frac{1}{n}\big(\frac{c_0}{r}-1\big)$.
\end{theorem}
Note that $c_0>r$, since $c_0^{-1}=\int_0^\infty e^{-rt}{}_tp_x\,dt<\int_0^\infty e^{-rt}\,dt=r^{-1}$.

\subsection{Subjective Mortality}

Suppose now that there are two survival functions, an objective one ${}_tp_x$ for the general population, and a subjective one ${}_t\tilde p_x$ for the individual in question. From an economic perspective, this is equivalent to having a subjective discount rate which differs from the risk free rate $r$. Naturally the individual would prefer an annuity with a different payout structure, because the Euler-Lagrange equations \eqref{lagrangecondition} become:
\begin{equation}
e^{-rt}{}_t\tilde p_x\,u'(c(t))=\lambda e^{-rt}{}_tp_x \quad\text{for every $t$.}
\end{equation}

In other words, the optimal annuity structure, under CRRA utility, is
\begin{equation}
c_S(t)=c_S(0)\Big(\frac{{}_t\tilde p_x}{{}_tp_x}\Big)^{1/\gamma}.
\end{equation}
So a healthier-than-normal individual (${}_t\tilde p_x>{}_tp_x$) will prefer back-loading the annuity stream. In essence, the mortality outlier finds annuity payments underpriced when they occur late in life, so would prefer to invest in an annuity that enhances those payments. Of course, this individual is going to be out of luck - annuities will be designed around the population's mortality, not theirs. But, this does raise a question about how one would test the extent to which individuals believe themselves to be healthier than the population: if different annuity structures are available, people who believe themselves to be healthier than the average annuity purchaser should signal that belief by choosing back-loaded annuity payouts. 

We can look at tontines in the same way (and test whether tontines might act as a proxy for back-loaded annuities). If $\tilde p$ and $p$ denote the subjective and objective survival probabilities, then Euler-Lagrange for CRRA utility is that the optimal payout rate $D_{n,\gamma}(p,\tilde p)$ satisfies
\begin{equation}
\tilde p D_{n,\gamma}(p,\tilde p)^{-\gamma}\theta_{n,\gamma}(p)=\lambda.
\end{equation}
In other words, if $\beta_{n,\gamma}(p,\tilde p)=\tilde p\theta_{n,\gamma}(p)$ then 
\begin{equation}
D_{n,\gamma}(p,\tilde p)=D_{n,\gamma}(1,1)\beta_{n,\gamma}(p,\tilde p)^{1/\gamma},
\end{equation}
where the specific mortality model again enters only via
\begin{equation}
D_{n,\gamma}(1,1)=\Big[\int_0^\infty e^{-rt}\beta_{n,\gamma}({}_tp_x, {}_t\tilde p_x)^{1/\gamma}\,dt\Big]^{-1}.
\end{equation}
What this also implies is that a tontine provides relatively more utility if the (participant) believes that the objective population survival rate is much lower the his or her individual survival rate. This, again, might explain the fact that some participants select the tontine -- perhaps they believe they are much healthier than the rest of the investment pool -- while some select the annuity. In ongoing work, we are attempting to quantify this effect.
\subsection{Items Ignored}

There are two issues we have not had a chance to address, which we leave for possible future research, and that is the role of \emph{credit risk} as well as the impact of \emph{stochastic mortality}. The existence of credit risk is a much greater concern for the buyers of life annuities, vs. tontines, given the risk taken on by the insurance sponsor. Likewise, under a stochastic mortality model the tontine payout would be more variable and uncertain, which might further reduce the utility of the tontine relative to the life annuity. We leave this for further investigation and focus now on the implications of our (simpler) model.

\section{Numerics and Examples}
\label{numerics}

\subsection{King William's tontine of 1639}

We start by displaying some basic simulation results and summary statistics for the actuarial present value (APV) of cash-flow from a tontine vs. a life annuity. In particular, we compute and display the actuarial (mortality adjusted) mean, standard deviation and skewness of the discounted cash-flows. The mechanics of the simulation itself are reported in the appendix and the parameter values we use coincide with the tontine and annuity payout rates described in section \ref{1693}. This is the easiest and most transparent way to compare the economic value of the tontine to the annuity. From this it is quite clear that the 14\% annuity was offering more -- in expected present value terms -- compared to the 10\% / 7\% tontine. 

\begin{center}
{\bf Table \ref{table04} and Table \ref{table05} go here}
\end{center}

For example, under a 6\% interest rate and (Gompertz smoothed) mortality rates the actuarial present value (APV) of a tontine paying 10\% for seven years and 7\% thereafter, for an $x=10$ year-old nominee is approximately \pounds 133. This is a 33\% premium to the \pounds 100 cost. This assumes a mortality basis as reported by Finlaison (1829). If we assume a higher mortality rate, with a modal value of life $m=50$, then the APV of the tontine is (an even lower) at \pounds 130, albeit with a much higher skewness of $11.16$.

\emph{Ceteris paribus}, at higher ages, higher interest rates, or more aggressive mortality assumptions, the APV is lower.  Notice that the skewness is (positive) and higher under more aggressive mortality assumption, namely when the force of mortality (proxied by $m$) is higher. In contrast to the tontine, under the same 6\% interest rate assumption the actuarial present value (APV) of a life annuity paying 14\% to an $x=10$ year-old nominee is equal to approximately \pounds 185. This is consistent with Halley's (1694) claim that the 14\% life annuity was worth almost 14 years purchase at age ten. Notice also that the skewness for the life annuity is always negative. So, for those investors who might value skewness (all else being equal) the tontine might be preferred to a life annuity, although the present value is clearly lower.

\subsection{Optimal Tontine in the 21st Century}

Figure \ref{fig2} displays the range of possible 4\% tontine dividends over time, assuming an initial pool of $n=400$ nominees, under a Gompertz law of mortality with parameters $m=88.721$ and $b=10$. This mortality basis corresponds to a survival probability of ${}_{35}p_{65}=0.05$, i.e. from age 65 to age 100 and is the baseline values for several of our numerical examples. The figure clearly shows an increasing payment stream conditional on survival, which isn't the optimal function. 

Indeed, in such a (traditional, historical) tontine scheme, the initial expected payment is quite low, relative to what a life annuity might offer in the early years of retirement, while the payment in the final years -- for those fortunate enough to survive -- would be very high and quite variable. It is not surprising then that this form of tontine is both suboptimal from an economic utility point of view and isn't very appealing to individuals who want to maximize their standard of living over their entire retirement years.

\begin{center}
{\bf Figure \ref{fig2} and Figure \ref{fig3} go here}
\end{center}

In contrast to Figure \ref{fig2} which are the sub-optimal tontine, Figure \ref{fig3} displays the range of outcomes from the optimal tontine payout function, under the same interest rate and mortality basis. To be very specific, Figure \ref{fig3} is computed by solving for the value of $D_{n,\gamma}^{\text{OT}}(1)$ and then constructing $D_{n,\gamma}^{\text{OT}}(_tp_x)$ for $n=400, r=0.04$ and $\gamma=1$.  Once the payout function is known for all $t$, the number of survivors at the 10th and 90th percentile of the Binomial distribution is used to bracket the range of the payout from age 65 to age 100. Clearly, the expected payout per survivor is relatively constant over the retirement years, which is much more appealing intuitively. Moreover, the discounted expected utility from a tontine payout such as the one displayed in Figure \ref{fig3} is much higher than the utility of the one displayed in Figure \ref{fig2}. 

\begin{center}
{\bf Table \ref{table06} goes here}
\end{center}

Table \ref{table06} displays the optimal tontine payout function for a very small pool of size $N=25$. These correspond to the $D_{n,\gamma}^{\text{OT}}(_tp_x)$ values derived in section \ref{theory}. Notice how the optimal tontine payout function is quite similar (identical in the first significant digit) regardless of the individual's Longevity Risk Aversion $\gamma$, even when the tontine pool is relatively small at $n=25$. The minimum guaranteed dividend starts-off at about 7\% at age 65 and then declines to approximately 1\% at age 95. Of course, the actual cash-flow payout to an individual, conditional on being alive does not necessarily decline and actually stays relatively constant. 

\begin{center}
{\bf Table \ref{table07} goes here}
\end{center}

Table \ref{table07} displays the utility indifference values for a participant at age 60. Notice how even a retiree with a very high level of Longevity Risk Aversion (LoRA) $\gamma$, will select a tontine (with pool size $n \geq 20$) instead of a life annuity if the insurance loading is greater than 7.5\%

\begin{center}
{\bf Table \ref{table08} goes here}
\end{center}

Finally, \ref{table08} computes the certainty equivalent factors. If an individual with LoRA $\gamma \neq 1$ is faced with a tontine structure that is only (optimal) for someone with LoRA $\gamma=1$ (i.e. logarithmic utility) the welfare loss is minuscule. This is why we advocate the \emph{natural tontine} payout function, which is only optimal for $\gamma=1$, as the basis for 21st century tontines. 

\begin{center}
{\bf Figure \ref{fig4} goes here}
\end{center}

Figure \ref{fig4} shows the difference between the optimal tontine payout function for different levels of Longevity Risk Aversion $\gamma$ is barely noticeable when the tontine pool size is greater than $N=250$, mainly due to the effect of the law of large numbers. This curve traces the \emph{minimum} dividend that a survivor can expect to receive at various ages. The median is (obviously) much higher.

\begin{center}
{\bf Figure \ref{fig5} goes here}
\end{center}

Figure \ref{fig5} illustrates the optimal tontine payout function for someone with logarithmic $\gamma=1$ utility starts-of paying the exact same rate as a life annuity regardless of the number of participants in the tontine pool. But, for higher levels of longevity risk aversion $\gamma$ and a relatively smaller tontine pool, the function starts-off at a lower value and declines at a slower rate

\begin{center}
{\bf Figure \ref{fig6} goes here}
\end{center}

Figure \ref{fig6} show that for a relatively smaller tontine pool size, the retiree who is highly averse to longevity risk $\gamma=25$ will want a guaranteed minimum payout rate (GMPR) at advanced ages that is higher than his or her projected survival probability. In exchange they will accept lower GMPR at lower ages. In contrast, the logarithmic utility maximizer will select a GMPR that is exactly equal to the projected survival probability

\begin{center}
{\bf Figure \ref{fig7} goes here}
\end{center}

Figure \ref{fig7} indicates than an actuarially fair life annuity that guarantees 7.5\% for life starting at age 65 provides more utility than an optimal tontine regardless of Longevity Risk Aversion (LoRA) or the size of the tontine pool. But, once an insurance loading is included, driving the yield under the initial payout from the optimal tontine, the utility of the life annuity might be lower. The indifference loading is $\delta$ and reported in table \ref{table07}.

So here is our main takeaway and idea in the paper, once again. The historical tontine in which dividends to the entire pool are a constant (e.g. 4\%) interest rate over the entire retirement horizon are suboptimal because they create an increasing consumption profile that is both variable and undesirable. However, a tontine scheme in which interest payments to the pool early-on are higher (e.g. 8\%) and then decline over time, so that the few winning centenarians receive a much lower interest rate (e.g. 1\%) is in fact the optimal policy. Coincidently, King William's 1693 tontine had a similar declining structure of interest payments to the pool, which was quite rare historically. 

We are careful to distinguish between the guaranteed \emph{interest} rate (e.g. 8\% or 1\%) paid to the entire pool, and the expected \emph{dividend} to the individual investor in the optimal tontine, which will be relatively constant over time, as is evident from Figure \ref{fig2}. Of course, the present value of the interest paid to the entire pool over time is exactly equal to the original contribution made by the pool itself. We are simply re-arranging and parsing cash-flows of identical present value, in a different manner over time.

We have also shown that the utility loss from a properly designed tontine scheme is quite small when compared to an actuarially fair life annuity, which is the work-horse of the pension economics and lifecycle literature. In fact, the utility of from a tontine might actually be higher than the utility generated by a pure life annuity when the insurance (commission, capital cost, etc.) loading exceeds 10\%. This result should not negate or be viewed as conflicting with the wide-ranging annuity literature which proves the optimality of life annuities in a lifecycle model. In fact, what we show is that it is still optimal to fully hedge longevity risk, but the instrument that one uses to do so depends on the relative costs. In other words, \emph{sharing} longevity risk amongst a relatively small ($n \leq 100$) pool of people doesn't create the large dis-utilities or welfare losses, at least within a classical rational model. This can also be viewed as a further endorsement of the participating life annuity, which lies in between the tontine and the conventional life annuity.

\section{Brief Literature Review}
\label{lit}

This is not the place -- not do we have the space -- for a full review of the literature on tontines, so we provide a selected list of key articles for those interested in further reading. 

The original tontine proposal by Lorenzo Tonti appears in French in Tontine (1654) and was translated and published in English in the wonderful collection of key historical actuarial articles edited by Haberman and Sibbett (1995). The review article by Kopf (1927) and the book by O'Donnell (1936) are quite dated, but do a wonderful job of documenting how the historical tontine operated, discussing its checkered history, and providing a readable biography of some of its earliest promoters in Denmark, Holland, France and England. 

The monograph by Cooper (1972) is devoted entirely to tontines and the foundations of the 19th century (U.S.) tontine insurance industry, which is based on the tontine concept but is somewhat different because of the savings and lapsation component. In a widely cited article, Ransom and Sutch (1987) provide the background and story of the banning of tontine insurance in New York State, and then eventually the entire U.S. The comprehensive monograph by Jennings and Trout (1982) reviews the history of tontines, with particular emphasis on the French period, while carefully documenting payout rates and yields from most known tontines. It is a minor classic in the field and is a primary source for anyone interested in tontines.

For those interested in the pricing of mortality-contingent claims during the 17th and 18th century, as well as the history and background of the people involved, we recommend Alter (1983, 1986), Poitras (2000), Hald (2003), Poterba (2005), Ciecka (2008a, 2008b), Rothschild (2009) as well as Bellhouse (2011), and of course, Homer and Sylla (2005) for the relevant interest rates. Another important reference within the history of mortality-contingent claim pricing is Finlaison (1829) who was the first to document the mortality experience of King William's 1693 tontine participants and argued that life annuities sold by the British government -- at the same price for all ages! -- were severely underpriced. 

More recently, the newspaper article by Chancellor (2001), the book by Lewin (2003) and especially the recent review by McKeever (2009) all provide a very good history of tontines and discuss the possibility of a tontine revival. The standard actuarial textbooks, such as Promislow (2011) or Pitacco, et. al. (2009) for example, have a few pages devoted to the historical tontine principal. More relevantly, a series of papers on pooled annuity funds, for example Piggot, Valdez and Detzel (2005) have attempted to reintroduce tontine-like structures. The closest paper we can find related to our \emph{natural} tontine payout function is the work by Sabin (2010) on a fair annuity tontine, as well as the recent paper by Donnelly, Guillen and Nielsen (2013) which derives an expression for the amount individual are willing to pay to avoid incurring longevity risk.

In sum, although the published research on tontine schemes is scattered across the insurance, actuarial, economic, and history journals, we have come across few, if any, scholarly articles that condemn or dismiss the tontine concept outright.  

\section{Conclusion and Relevance}
\label{conc}

It is not widely known that in the year 1790, the first U.S. Secretary of the Treasury, Alexander Hamilton proposed what some have called the nation's first (financial) Hunger Game. To help reduce a crushing national debt -- something that is clearly not a recent phenomenon -- he suggested the U.S. government replace high-interest revolutionary war debt with new bonds in which coupon payments would be made to a group as opposed to individuals\footnote{Source: Jennings, Swanson and Trout (1988)}. The group members would share the interest payments evenly amongst themselves, provided they were alive. But, once a member of the group died, his or her portion would stay in a pool and be shared amongst the survivors. This process would theoretically continue until the very last survivor would be entitled to the entire -- potentially millions of dollars -- interest payment. This obscure episode in U.S. history has become known as Hamilton's Tontine Proposal, which he claimed -- in a letter to George Washington -- would reduce the interest paid on U.S. debt, and eventually eliminate it entirely.

Although Congress decided to pass on Hamilton's proposal (Hamilton left public life in disgrace after admitting to an affair with a married woman, and soon-after died in duel with the U.S. vice president at the time Aaron Burr) the tontine idea never died. 

U.S. insurance companies began issuing tontine insurance policies -- which are close cousins to Tonti's tontine -- to the public in the mid-19th century, which became wildly popular\footnote{Source: Ransom and Sutch (1987)}. By the start of the 20th century, historians have documented that half of U.S. households owned a tontine insurance policy, which many used to support themselves through retirement. The longer one lived, the greater their payments. This was a personal hedge against longevity, with little risk exposure for the insurance company. Sadly though, due to shenanigans and malfeasance on the part of company executives, the influential New York State Insurance Commission banned tontine insurance in the state, and by 1910 most other states followed. Tontines have been illegal in the U.S. for over a century and most insurance executives have likely never heard of them.

In sum, tontines not-only have a fascinating history but are actually based on sound economic principles, In fact, Adam Smith himself, quoted at the beginning of this paper, noted that tontines are preferred to life annuities. We believe that a strong case can be made for overturning the current ban on tontine insurance -- \emph{allowing both tontine and annuities to co-exist as they did 320 years ago} -- with suitable adjustments to alleviate problems encountered in the early 20th century. Indeed, given the insurance industry's concern for longevity risk capacity, and its poor experience in managing the risk of long-dated fixed guarantees, one can argue that an (optimal) tontine annuity is a triple win proposition for individuals, corporations and governments. 

It is worth noting that under the proposed (EU) Solvency II guidelines for insurer's capital as well as risk-management, there is a renewed focus on {\em total balance sheet} risks. In particular, insurers will be required to hold {\em more} capital against market risk, credit risk and operational risk. In fact, in a recently released report by Moody's Investor Services\footnote{Source: \emph{European Insurers: Solvency II - Volatility of Regulatory Ratios Could Have Broad Implications For European Insurers}, May 2013, available on www.moodys.com}, they claim that \emph{solvency ratios will exhibit a more complex volatility under Solvency II than under Solvency I, as both the available capital and the capital requirements will change with market conditions.} According to many commentators this is likely to translate into higher consumer prices for insurance products with long-term maturities and guarantees. And, although this only applies to European companies (at this point time), it is not unreasonable to conclude that in a global market annuity loadings will increase, making (participating) tontine products relatively more appealing to price-sensitive consumers.

Moreover, perhaps a properly designed tontine product could help alleviate the low levels of voluntary annuity purchases -- a.k.a. the annuity puzzle -- by gaming the behavioral biases and prejudices exhibited by retirees. The behavioral economics literature and advocates of cumulative prospect theory have argued that consumers make decisions based on more general \emph{value functions} with personalized decisions weights. Among other testable hypotheses, this leads to a preference for investments with (highly) skewed outcomes, even when the alternative is a product with the same expected present values\footnote{Source: Barberis (2013)}. Our simulations indicate that the skewness (third moment) of tontine payouts is positive and higher than the skewness of a life annuity, especially if purchased at older ages. We believe this is yet another argument in favor of re-introducing tontines. Of course, whether the public and regulators can be convinced of these benefits remains to be seen, but a debate would be informative.

We are not alone in this quest. Indeed, during the last decade a number of companies around the world -- egged on by scholars and journalists\footnote{See for example: Chancellor (2001), Goldsticker (2007), Pechter (2007), Chung and Tett (2007), as well as the more scholarly articles by Richter and Weber (2011), Rotemberg (2009) and especially Sabin's (2010) Fair Tontine Annuity.} -- have tried to resuscitate the tontine concept (with patents pending of course) while trying to avoid the bans and taints. Although the specific designs differ from proposal to proposal, all share the same idea we described in this paper: Companies act as custodians and guarantee very little. This arrangement requires less capital which then translates into more affordable pricing for the consumer. Once again our models indicate that a properly designed tontine could hold its own against an actuarially fair life annuity and pose a real challenge to a loaded annuity. 

We suspect that the biggest obstacle to bringing back tontines is the name itself and the image it evokes. Might we conclude by suggesting the weightier \emph{Hamiltonian} as an alternative? Each share would be offered for a modest \$10 investment.

\newpage

\section{Appendix}
\label{apen}

\subsection{Proofs}
\label{proof1}

To simplify notation, we write $\theta(p)=\theta_{n,\gamma}(p)$, $\beta(p)=\beta_{n,\gamma}(p)$, and $D(p)=D^{\text{OT}}_{n,\gamma}(p)$. 

\begin{proof}[Proof of Lemma \ref{betabound}]
The case $\gamma=1$ is trivial. For $2\le\gamma$ there is a simple proof, since
$$
\theta(p)=E\Big[\Big(\frac{N(p)}{n}\Big)^{\gamma-1}\Big]\ge E\Big[\frac{N(p)}{n}\Big]^{\gamma-1}=\Big(\frac{1+(n-1)p}{n}\Big)^{\gamma-1}=\Big(p+\frac{1-p}{n}\Big)^{\gamma-1}
>p^{\gamma-1}
$$ 
by Jensen, implying $\beta(p)=p\theta(p)>p^\gamma$. To prove the general case $\gamma>0$ requires a more involved argument, based on the following calculation:
\begin{lemma}
$E[\frac{n}{N(p)}]<\frac{1}{p}$. 
\label{reciprocalbound}
\end{lemma}
\begin{proof}
\begin{align*}
E\Big[\frac{n}{N(p)}\Big]&=\sum_{k=0}^{n-1}\binom{n-1}{k}p^k(1-p)^{n-1-k}\frac{n}{k+1}\\
&=\frac{1}{p}\sum_{k=0}^{n-1}\binom{n}{k+1}p^{k+1}(1-p)^{n-(k+1)}
=\frac{1}{p}[1-(1-p)^n]<\frac{1}{p}.
\end{align*}
\end{proof}
Suppose first that $0<\gamma<1$. 
H\"older's inequality implies that $E[(\frac{n}{N(p)})^a]^{1/a}$ is increasing in $a>0$.  
In particular, whenever $0<a<1$, we have 
\begin{equation}
E\Big[\Big(\frac{n}{N(p)}\Big)^a\Big]^{1/a}<E\Big[\frac{n}{N(p)}\Big]<\frac{1}{p},
\label{inequalityfora}
\end{equation}
by Lemma \ref{reciprocalbound}. 
Taking $a=1-\gamma$, we obtain that 
$\theta(p)<(\frac{1}{p})^{1-\gamma}=p^{\gamma-1}$. Therefore $\beta(p)=p\theta(p)<p^\gamma$ as required. 

Now suppose $\gamma>1$. Set 
$A=\lim_{a\downarrow 0}E[(\frac{n}{N(p)})^a]^{1/a}<\frac{1}{p}$ by \eqref{inequalityfora}. 
By l'H\^opital's rule,
$$
A=e^{\lim_{a\downarrow 0}\frac{1}{a}\log E[(\frac{n}{N(p)})^a]}
=e^{\lim_{a\downarrow 0}E[(\frac{n}{N(p)})^a\log(\frac{n}{N(p)})]/E[(\frac{n}{N(p)})^a]}
=e^{-E[\log(\frac{N(p)}{n})]}.
$$
Taking log's, we obtain that 
\begin{equation}
E\Big[\log\big(\frac{N(p)}{n}\big)\Big]>\log p.
\label{loginequality}
\end{equation}
As above,
$E[(\frac{N(p)}{n})^a]^{1/a}$ is increasing in $a>0$, so by \eqref{loginequality} and the same l'H\^opital argument as before, 
$$
E\Big[\Big(\frac{N(p)}{n}\Big)^{\gamma-1}\Big]^{\frac{1}{\gamma-1}}>
\lim_{a\downarrow 0}E\Big[\Big(\frac{N(p)}{n}\Big)^a\Big]^{1/a}
=e^{E[\log(\frac{N(p)}{n})]}>e^{\log p}=p.
$$
Therefore $\theta(p)>p^{\gamma-1}$, so $\beta(p)=p\theta(p)>p^\gamma$, proving Lemma \ref{betabound}. 
\end{proof}

\begin{proof}[Proof of Theorem \ref{utilityinequality}]
The case $\gamma\neq 1$ follows immediately from Lemma \ref{betabound}. When $\gamma=1$, we have
$u(c)=\log c$ so $U_{n,1}^{\text{OT}}=\int_0^\infty e^{-rt}{}_tp_x \,E[\log(\frac{nc_0\,{}_tp_x}{N({}_tp_x)})]\,dt$, while 
$U_1^{\text{A}}=\int_0^\infty e^{-rt}{}_tp_x \,\log(c_0)\,dt$. Therefore 
$$
U_1^{\text{A}}-U_{n,1}^{\text{OT}}=\int_0^\infty e^{-rt}{}_tp_x \Big(E\Big[\log\big(\frac{N({}_tp_x)}{n}\big)\Big]-\log({}_tp_x)\Big)\,dt,
$$
and the result now follows by \eqref{loginequality}.

\end{proof}

\subsection{Indifference loading}
\begin{proof}[Proof of Theorem \ref{loadinginequality}]
Since $1<\gamma\le 2$, $c^{\gamma-1}$ is concave in $c$, so lies below its tangents. Therefore
$$
\theta(p)=E\Big[\Big(\frac{N(p)}{n}\Big)^{\gamma-1}\Big]
\le E\Big[p^{\gamma-1}+(\gamma-1)p^{\gamma-2}\Big(\frac{N}{n}-p\Big)\Big]
=p^{\gamma-1}+(\gamma-1)p^{\gamma-2}\frac{1-p}{n}.
$$
$c^{1/\gamma}$ is also strictly concave in $c$, so in the same way
$$
\beta(p)^{\frac{1}{\gamma}}=(p\theta(p))^{\frac{1}{\gamma}}
<(p^\gamma)^{\frac{1}{\gamma}}+\frac{1}{\gamma}(p^\gamma)^{\frac{1}{\gamma}-1}\cdot (\gamma-1)p^{\gamma-1}\frac{1-p}{n}
=p+\frac{(\gamma-1)(1-p)}{\gamma n}.
$$
Therefore 
$$
\int_0^\infty e^{-rt}\beta({}_tp_x)^{\frac{1}{\gamma}}\,dt < \frac{1}{c_0} + \frac{\gamma-1}{\gamma n}\Big(\frac{1}{r}-\frac{1}{c_0}\Big).
$$
By definition, 
$$
\frac{c_0^{-\gamma}}{1-\gamma}(1-\delta)^{1-\gamma}
=\frac{1}{1-\gamma}\Big(\int_0^\infty e^{-rt}\beta({}_tp_x)^{1/\gamma}\\,dt\Big)^\gamma,
$$
and $1-c^{\frac{\gamma}{1-\gamma}}$ is also concave in $c$, so as before 
$$
\delta=1-\Big(c_0\int_0^\infty e^{-rt}\beta({}_tp_x)^{\frac{1}{\gamma}}\,dt\Big)^{\frac{\gamma}{1-\gamma}}\cdot 
<-\frac{\gamma}{1-\gamma}\cdot \frac{\gamma-1}{\gamma n}\Big(\frac{c_0}{r}-1\Big)=\frac{1}{n}\Big(\frac{c_0}{r}-1\Big)
$$
as required.
\end{proof}
For any $\gamma$ one can derive (using the second moment of $N(p)$ now, as well as the first) the asymptotic result that $\delta n\to \frac{\gamma}{2}(\frac{c_o}{r}-1)$. But convergence turns out to be so slow that this precise asymptotic is of limited use. (The slow convergence derives from the observation that, with Gompertz mortality, the time $t$ till ${}_tp_x$ reaches $\frac{1}{n}$ grows only at rate $b\log\log n$.) For example, with $\gamma=2$, $r=3\%$, age $x=50$, and Gompertz parameters $m=87.25$ and $b=9.5$, we obtain $\frac{\gamma}{2}(\frac{c_o}{r}-1)=0.6593$; But for $n=$ 10, 100, or 1000 we only have $\delta n=$ 0.2858, 0.3377, and 
0.3671; In fact, even with $n=7\times 10^9$ tontine participants (roughly the world's entire population, all postulated to share the same age and hazard rate), we would only reach $\delta n = 0.4417$

\subsection{Certainty Equivalents, and the Natural Tontine}
We wish to compare the welfare loss experienced by an individual with longevity risk aversion $\gamma\neq 1$, if they participate in a natural tontine rather than an optimal one. We therefore calculate the ratio $\Gamma\ge 1$ of the certainty equivalents of the two tontines. This represents the initial deposit into a natural tontine needed to provide the same utility as a \$1 deposit into an optimal one.
The natural tontine has utility 
$$
U_{n,\gamma}^{\text{N}}=\frac{c_0^{1-\gamma}}{1-\gamma}\int_0^\infty e^{-rt}{}_tp_x^{2-\gamma}\theta_{n,\gamma}({}_tp_x)\,dt.
$$
Therefore
$$
\Gamma
=\left[\frac{U^{\text{OT}}_{n,\gamma}}{U^{\text{N}}_{n,\gamma}}\right]^{\frac{1}{1-\gamma}}
=\left[\frac{D^{\text{OT}}_{n,\gamma}(1)^{-\gamma}}{c_0^{1-\gamma}\int_0^\infty e^{-rt}{}_tp_x^{2-\gamma}\theta_{n,\gamma}({}_tp_x)\,dt}\right]^{\frac{1}{1-\gamma}}
$$
which is then leads to:
$$
\Gamma=\frac{\big(\int_0^\infty e^{-rt}{}_tp_x\,dt\big)\big(\int_0^\infty e^{-rt}\beta_{n,\gamma}({}_tp_x)^{\frac{1}{\gamma}}\,dt)^{\frac{\gamma}{1-\gamma}}}{\big(\int_0^\infty e^{-rt}{}_tp_x^{2-\gamma}\theta_{n,\gamma}({}_tp_x)\,dt\big)^{\frac{1}{1-\gamma}}}.
$$

When we compute these values numericaly, for $0<\gamma\le 2$, we get values quite close to 1. Moreover, if we let 
$n\to\infty$, then 
$\beta_{n,\gamma}(p)^{1/\gamma}\to p$ and $\theta_{n,\gamma}(p)\to p^{-(1-\gamma)}$, which makes $\Gamma\to 1$ asymptotically, as long as $\gamma\le 2$.

For $\gamma>2$ this is not a fair comparison, as the integral diverges and $U_{n,\gamma}^{\text{N}}=\infty$. It is worth understanding why. For the optimal tontine, the integrand involves $[p\theta_{n,\gamma}(p)]^{1/\gamma}\sim p^{\frac{1}{\gamma}}n^{\frac{1-\gamma}{\gamma}}$ ie. is well behaved as $p\to 0$. On the other hand, the integral for the natural tontine involves $p^{2-\gamma}\theta_{n,\gamma}(p)$ which can be large when $p\to 0$ if $\gamma>2$. This means that for $\gamma>2$ the natural tontine utility is unduly influenced by the possibility of living to highly advanced ages. Even if there is only a single survivor, that survivor's payout will have dropped to quite low levels by the time it is actually improbable that anyone will live that long. And for $\gamma>2$ the negative consequences of the low payout dominate the small probability of surviving that long. 

\subsection{Description of Tontine Simulation}

We start with a pool of $n=1000$ (for example) homogenous individuals who assume the role of both the annuitant and nominee, each of age $x$ with maximum lifespan $\omega=105$, at which point everyone is dead. These $n$ individuals contribute $\pounds w$ to the tontine pool for a total of $\pounds wn$. We use the symbol triplet $L(x,i,t)$ to denote the life state of the $i^{th}$ individual in year $t=\omega-x$. Formally, $L(x,i,t)=1$ while the $i^{th}$ nominee is still alive and $L(x,i,t)=0$ once the nominee is dead. We start with $L(x,i,0)=1$. The next value of $L(x,i,t)$ is obtained by simulating a standard uniform [0,1] �killing� random variable $u$, and then setting the variable $L(x,i,t)=0$ whenever $u<q_{x+t}$ where $q_{x+t}$ is the mortality rate applicable at age $(x+t)$. So, for example, if the mortality rate at age 30 is $q_{30} = 0.15$, and the random variable outcome for $u = 0.2$, then the individual survives. But, if $u = 0.1$, the individual is killed and all future $L(x,i,t)$ values are set to zero. (No resurrections allowed!) Obviously, the greater the value of $q_{x+t}$, the higher the probability (and realization) of death.  This might not be the most efficient or fastest way to simulate the matrix of lifespans, but it's the easiest to explain. 

In sum, for each simulation run denoted by $j$, this process generates a (big) matrix $L(j)$ with $n$ rows and $(\omega-x+1)$ columns. The first row is set to be all ones -- everyone starts alive -- and all rows slowly decay to zeros over time. Finally, the last column is (forced) to be all zeros. We set the baseline to be M = 10,000 simulations so that: $1\leq j \leq 10000$, and the initial age will be $x=10$, which was the average age of the nominees in King William's tontine of 1693.

At the end of each year a total of $\pounds wnd(t)$ is distributed as a dividend to the survivors, where d(t) denotes the value of the tontine payout function at time $t$. For the 1693 tontine for example, $d(t)  = 10\%$ for the first 7 years (to the year 1700), and then $d(t) = 7\%$ thereafter. This then generates a new matrix $D(x,i,t)$ which denotes the cash dividend to the $i^{th}$ individual in the $t^{th}$ year. Note that $D(x,i,t) = 0$, whenever $L(x,i,t)=0$. Dead people share no dividends. The process of computing $D(x,i,t)$ is rather simple, assuming nominee $(x,i)$ is alive in time period $t.$ Namely, divide $wnd(t)$ which is the total interest payable to the surviving pool members, by the number of people alive, which is sum of $L(x,i,t)$ from $i=1$ to $i = n$. Formally it is:

$$
D(x,i,t) = \frac{wnd(t)}{\sum_{i=1}^N L(x,i,t)}
$$

Now, the main quantity we interested in computing is the variable $PV(x,i,j)$, which is the present value of the tontine payout to nominee $(x,i)$, in simulation run number $j \leq M$. Formally it is defined as:

$$
PV(x,i,j) = \sum_{t=1}^{\omega-x} \frac{D(x,i,t)}{(1+R)^t}
$$
where R is the valuation rate (assumed to be 6\% for most cases). Algorithmically, the entire numerator vector is simulated -- first, based on the number of other survivors -- and then discounted to arrive at a present value for the tontine payoff to nominee $(x,i)$, in one particular simulation run. Notice that once $L(x,i,t)$ is zero (the nominee is dead), then $D(x,i,t)$ is zero as well so the entire summation is valid. Every simulation run $j$, will generate a value for $PV(x,i,j)$ for a total of $M = 10,000$ present values for a given representative nominee $(x,i).$ 

Finally, we are interested in the sample mean, standard deviation, skewness and kurtosis of these $M = 10,000$ (simulated) present values. For further clarity of notation, note that the (simulation) sample mean, which is an actuarial present value, is defined as:

$$
APV(x,i) = \frac{1}{M} \sum_{j=1}^M PV(x,i,j)
$$

The (simulation) sample standard deviation is:

$$
SDPV(x,i) = \sqrt{\frac{1}{M} \sum (PV(x,i,j)-APV(x,i))^2}
$$
The sample skewness and kurtosis is defined in a similar manner. 

\newpage

\begin{table}[h]
\begin{center}
\begin{tabular}{||c||c||c||c||}
\hline\hline
\multicolumn{4}{||c||}{\textbf{King William's Tontine of 1693}} \\ 
\hline\hline
& \textit{\# of Nominees} & \textit{\# of Shares} & \textit{Average Age at
Nomination} \\ \hline\hline
\textbf{Males } & 604 & 653 & 10.85 \\ \hline\hline
\textbf{Females} & 409 & 428 & 10.98 \\ \hline\hline
\textbf{Total} & 1,013 & 1,081 & 10.90 \\ \hline\hline
\multicolumn{4}{||c||}{\footnotesize Source: Raw data from Howard (1694). Compiled by The IFID Centre} \\ \hline\hline
\end{tabular}\medskip
\caption{1,013 nominees remained in the tontine, after the annuity conversion option expired. There was a single share class, paying a guaranteed dividend of \pounds 10 (semiannually) until June 1700 and then \pounds 7 thereafter, until seven nominees remained. Surprisingly, over 60\% of the nominees were male, while the last nominee (a female) died at the age of 100 in the year 1783 after receiving a tontine dividend of \pounds 1,081.  \label{table01}}
\smallskip
\end{center}
\end{table}

\begin{table}[h]
\begin{center}
\begin{tabular}{||c||c||c||}
\hline\hline
\multicolumn{3}{||c||}{\textbf{King William's Tontine of 1693}} \\ 
\hline\hline
\textbf{\# of Nominees} & \textbf{Shares Purchased} & \textbf{Total Shares}
\\ \hline\hline
\textbf{956} & 1 & 956 \\ \hline\hline
\textbf{51} & 2 & 102 \\ \hline\hline
\textbf{3} & 3 & 9 \\ \hline\hline
\textbf{1} & 4 & 4 \\ \hline\hline
\textbf{2} & 5 & 10 \\ \hline\hline
\textbf{1,013} &  & 1,081 \\ \hline\hline
\multicolumn{3}{||c||}{\footnotesize Source: Raw data from Howard (1694). Compiled by The IFID Centre } \\ \hline\hline
\end{tabular}\medskip
\caption{The gap between the number of nominees (1013) and the number of tontine shares (1081) has caused some confusion over the years. Various sources have reported different values for the size of the original tontine pool. It is also worth noting that the average age of nominees with multiple shares (13.7 years) was higher than the average age of one-share nominees (10.7 years) by three years. So in fact, nominees with multiple shares contingent on their life died earlier (on average.) \label{table02}}
\end{center}
\end{table}

\begin{table}
\begin{center}
\begin{tabular}{||c||c||c||c||}
\hline\hline
\multicolumn{4}{||c||}{\textbf{King William's Tontine of 1693}} \\ 
\hline\hline
\textbf{Age} & \textbf{Number of Nominees} & \textbf{Alive in
1730} & \textbf{Alive in 1749} \\ \hline\hline
\textbf{0-2} & 99 = 50m + 49f & 58 = 25m + 30f & 34 = 20m + 14f \\ 
\hline\hline
\textbf{3-5} & 183 = 108m + 75f & 105 = 60m + 45f & 60 = 36m + 24f \\ 
\hline\hline
\textbf{6-8} & 174 = 111m + 63f & 93 = 53m + 40f & 53 = 30m + 23f \\ 
\hline\hline
\textbf{9-11} & 181 = 102m + 79f & 100 = 48m + 52f & 52 = 25m + 27f \\ 
\hline\hline
\textbf{12-14} & 138 = 82m + 56f & 63 = 34m + 29f & 35 = 17m + 18f \\ 
\hline\hline
\textbf{15-17} & 69 = 37m + 32f & 31 = 16m + 15f & 14 = 4m + 10f \\ 
\hline\hline
\textbf{18-20} & 50 = 25m + 25f & 22 = 11m + 11f & 7 = 3m + 4f \\ 
\hline\hline
\textbf{21-23} & 41 = 27m + 14f & 17 = 9m + 8f & 7 = 3m + 4f \\ \hline\hline
\textbf{24-26} & 22 = 14m + 8f & 10 = 6m + 4f & 5 = 3m + 2f \\ \hline\hline
\textbf{27-29} & 14 = 8m + 6f & 3 = 3m + 0f & 0 \\ \hline\hline
\textbf{30-32} & 16 = 10m + 6f & 6 = 4m + 2f & 0 \\ \hline\hline
\textbf{33-35} & 7 = 5m + 2f & 3 = 1m + 2f & 0 \\ \hline\hline
\textbf{36-38} & 7 = 4m + 3f & 1 = 1m + 0f & 0 \\ \hline\hline
\textbf{39-41} & 6 = 4m + 2f & 2 = 1m + 1f & 0 \\ \hline\hline
\textbf{42-44} & 1 = 1m + 0f & 0 & 0 \\ \hline\hline
\textbf{45-47} & 3 = 3m  + 0f & 0 & 0 \\ \hline\hline
\textbf{48-50} & 1 = 0m + 1f & 0 & 0 \\ \hline\hline
\textbf{51-53} & 1 = 0m + 1f & 0 & 0 \\ \hline\hline
\textbf{Avg. Age} & 11.10 & 9.83 (46.83) & 8.59 (64.59) \\ 
\hline\hline
\textbf{Total:} & 1013 & 514 & 267 \\ \hline\hline
\multicolumn{4}{||c||}{\footnotesize Source: Raw data from Howard (1694), Anonymous (1730) \& Anonymous (1749)} \\ \hline\hline
\end{tabular}\medskip
\caption{While it is quite natural that more infants weren't nominated, given high child mortality rates in the first few years of life, it is rather puzzling that so many tontine nominees were of such advanced ages. In fact, by the year 1749 none of the original nominees above the age of 26 were still alive. \label{table03}}
\end{center}
\end{table}

\begin{table}
\begin{center}
\begin{tabular}{||c||c||c||c||c||}
\hline\hline
\multicolumn{5}{||c||}{\textbf{King William's Tontine of 1693:
Simulated Results at Age 10 per \pounds 100}} \\ \hline\hline
\textbf{Interest} & \textbf{Mortality Basis: GoMa.} & \textbf{APV} & 
\textbf{SDev. PV} & \textbf{Skew. PV} \\ \hline\hline
\textbf{4\%} & $l=0.0104,m=69.5,b=13.8$ & \pounds 186.54 & \pounds %
96.05 & 4.06 \\ \hline\hline
\textbf{4\%} & $l=0, m=50,b=10$ & \pounds 174.91 & \pounds 96.98 & 13.18
\\ \hline\hline
\textbf{6\%} & $l=0.0104,m=69.5,b=13.8$ & \pounds 133.02 & \pounds %
45.74 & -0.46 \\ \hline\hline
\textbf{6\%} & $l=0, m=50,b=10$ & \pounds 130.31 & \pounds 48.5 & 11.16 \\ 
\hline\hline
\textbf{8\%} & $l=0.0104,m=69.5,b=13.8$ & \pounds 103.15 & \pounds %
28.88 & -1.73 \\ \hline\hline
\textbf{8\%} & $l=0, m=50,b=10$ & \pounds 102.10 & \pounds 21.41 & 2.42 \\ 
\hline\hline
\multicolumn{5}{||c||}{\footnotesize Notes: Based on simulation assumptions described in the appendix.} \\ \hline\hline
\end{tabular}\medskip
\caption{Under a 6\% interest rate and (Gompertz Makeham smoothed) mortality rates as reported in Finlaison (1829), the actuarial present value (APV) of a tontine paying 10\% for seven years and 7\% thereafter, for an $x=10$ year-old nominee is approximately \pounds 133. This is a 33\% premium to the cost.  At ages higher than $x=10$, higher interest rates or more aggressive mortality assumption, the APV is lower. \label{table04}}
\end{center}
\end{table}

\begin{table}
\begin{center}
\begin{tabular}{||c||c||c||c||c||}
\hline\hline
\multicolumn{5}{||c||}{\textbf{Life Annuity Paying 14\% Income: Simulated
Results at Age 10 per \pounds 100}} \\ \hline\hline
\textbf{Interest} & \textbf{Mortality Basis: GoMa.} & \textbf{APV} & 
\textbf{SDev. of APV} & \textbf{Skew. of APV} \\ \hline\hline
\textbf{4\%} & $l=0.0104,m=69.5,b=13.8$ & \pounds 244.05 & \pounds %
87.07 & -1.19 \\ \hline\hline
\textbf{4\%} & $l=0, m=50,b=10$ & \pounds 245.16 & \pounds 54.24 & -1.73
\\ \hline\hline
\textbf{6\%} & $l=0.0104,m=69.5,b=13.8$ & \pounds 184.53 & \pounds %
57.88 & -1.60 \\ \hline\hline
\textbf{6\%} & $l=0,m=50,b=10$ & \pounds 191.13 & \pounds 35.44 & -2.44
\\ \hline\hline
\textbf{8\%} & $l=0.0104,m=69.5,b=13.8$ & \pounds 147.55 & \pounds %
41.17 & -1.99 \\ \hline\hline
\textbf{8\%} & $l=0,m=50,b=10$ & \pounds 155.38 & \pounds 23.19 & -3.19
\\ \hline\hline
\multicolumn{5}{||c||}{\footnotesize Notes: Based on simulation assumptions described in the appendix} \\ \hline\hline
\end{tabular}\medskip
\caption{Under a 6\% interest rate and (Gompertz Makeham smoothed) mortality rates as reported in Finlaison (1829), the actuarial present value (APV) of a life annuity paying 14\% to an $x=10$ year-old nominee is worth approximately \pounds 185. This is consistent with Halley's (1694) claim that the 14\% life annuity was worth almost 14 years purchase at age ten. \label{table05}}
\end{center}
\end{table}

\begin{table}
\begin{center}
\begin{tabular}{||c||c||c||c||}
\hline\hline
\multicolumn{4}{||c||}{\textbf{Optimal Tontine Payout Function: Pool of Size}
$n=25$} \\ \hline\hline
LoRA ($\gamma $) & Payout Age 65 & Payout Age 80 & Payout
Age 95 \\ \hline\hline
0.5 & 7.565\% & 5.446\% & 1.200\% \\ \hline\hline
1.0 & 7.520\% & 5.435\% & 1.268\% \\ \hline\hline
1.5 & 7.482\% & 5.428\% & 1.324\% \\ \hline\hline
2.0 & 7.447\% & 5.423\% & 1.374\% \\ \hline\hline
4.0 & 7.324\% & 5.410\% & 1.541\% \\ \hline\hline
9.0 & 7.081\% & 5.394\% & 1.847\% \\ \hline\hline
\textbf{Survival} & ${}_0p_{65}=$100\% & ${}_{15}p_{65}=$72.2\% & ${}_{30}p_{65}=$16.8\% \\ \hline\hline
\multicolumn{4}{||c||}{\footnotesize Notes: Assumes $r=4\%$ and Gompertz Mortality ($m=88.72,b=10$)} \\ \hline\hline
\end{tabular}\medskip
\caption{Notice how the optimal tontine payout function $D_{n,\gamma}^{\text{OT}}({}_tp_x)$ is quite similar (identical in the first significant digit) regardless of the individual's Longevity Risk Aversion $\gamma$, even when the tontine pool is relatively small at $n=25$. The minimum guaranteed dividend starts-off at about 7\% at age 65 and then declines to approximately 1\% at age 95. Of course, the actual cash-flow received by a survivor does not necessarily decline and stays relatively constant. \label{table06}}
\end{center}
\end{table}

\begin{table}[h]
\begin{center}
\begin{tabular}{||c||c|c|c|c|c||}
\hline\hline
\multicolumn{6}{||c||}{{\bf The Highest Annuity Loading $\delta$ You Are Willing to Pay}} \\ \hline\hline
\multicolumn{6}{||c||}{{\bf If a Tontine Pool of Size $n$ is Available}} \\ \hline\hline
LoRA $\gamma$ & $n=20$ & $n=100$ & $n=500$ & $n=1000$ & $n=5000$\\ \hline\hline
0.5 & 72.6 b.p. & 14.5 b.p. & 2.97 b.p. & 1.50 b.p. & 0.30 b.p. \\ \hline
1.0 &   129.8 b.p. & 27.4 b.p. & 5.74 b.p. & 2.92 b.p. & 0.60 b.p.  \\ \hline
1.5 & 182.4 b.p. & 39.8 b.p. & 8.45 b.p. & 4.31 b.p. & 0.89 b.p. \\ \hline
2.0 & 231.7 b.p. & 51.8 b.p. & 11.1 b.p. & 5.68 b.p. & 1.18 b.p. \\ \hline
3.0 & 323.1 b.p. &  75.1 b.p. &  16.3 b.p. & 8.38 b.p. & 1.75 b.p. \\ \hline
9.0 & 753.6 b.p. &  199.8 b.p. &  45.9 b.p. &  23.8 b.p. &  5.09 b.p. \\ \hline
\hline
\multicolumn{6}{||c||}{\footnotesize Assumes Age $x=60$, $r=3\%$ and Gompertz Mortality ($m=87.25,b=9.5$)} \\ \hline\hline
\end{tabular}\medskip
\caption{Even a retiree with a very high level of Longevity Risk Aversion (LoRA) $\gamma=9$, will select a tontine (with pool size $n \geq 20$) instead of a life annuity if the insurance loading is greater than 7.5\%; Optimal tontine's aren't so bad!}
\label{table07}
\smallskip
\end{center}
\end{table}

\begin{table}[h]
\begin{center}
\begin{tabular}{||c||c|c|c||}
\hline\hline
\multicolumn{4}{||c||}{{\bf Natural vs. Optimal Tontine}} \\ \hline\hline
\multicolumn{4}{||c||}{{\bf Certainty Equivalent for $n=100$}} \\ \hline\hline
Age $x$ & $\gamma=0.5$ & $\gamma=1$ & $\gamma=2$  \\ \hline\hline
30  &  1.000018	 &  1	 &  1.000215 \\  \hline
40  &  1.000026	 &  1	 &  1.000753 \\  \hline
50  &  1.000041	 &  1	 &  1.001674 \\  \hline
60  &  1.000067	 &  1	 &  1.003388 \\  \hline
70  &  1.000118	 &  1	 &  1.003451 \\  \hline
80  &  1.000225	 &  1	 &  1.009877\\
\hline\hline
\multicolumn{4}{||c||}{\footnotesize $r=3\%$ and Gompertz $m=87.25,b=9.5$} \\ \hline\hline
\end{tabular}\medskip
\caption{If an individual with LoRA $\gamma \neq 1$ is faced with a tontine structure that is only (optimal) for someone with LoRA $\gamma=1$ (i.e. logarithmic utility) the welfare loss is minuscule. This is why we advocate the \emph{natural tontine} payout function, which is only optimal for $\gamma=1$ as the basis for 21st century tontines. \label{table08}}
\smallskip
\end{center}
\end{table}

\begin{figure}[here]
\begin{center}
\includegraphics[width=1.0\textwidth]{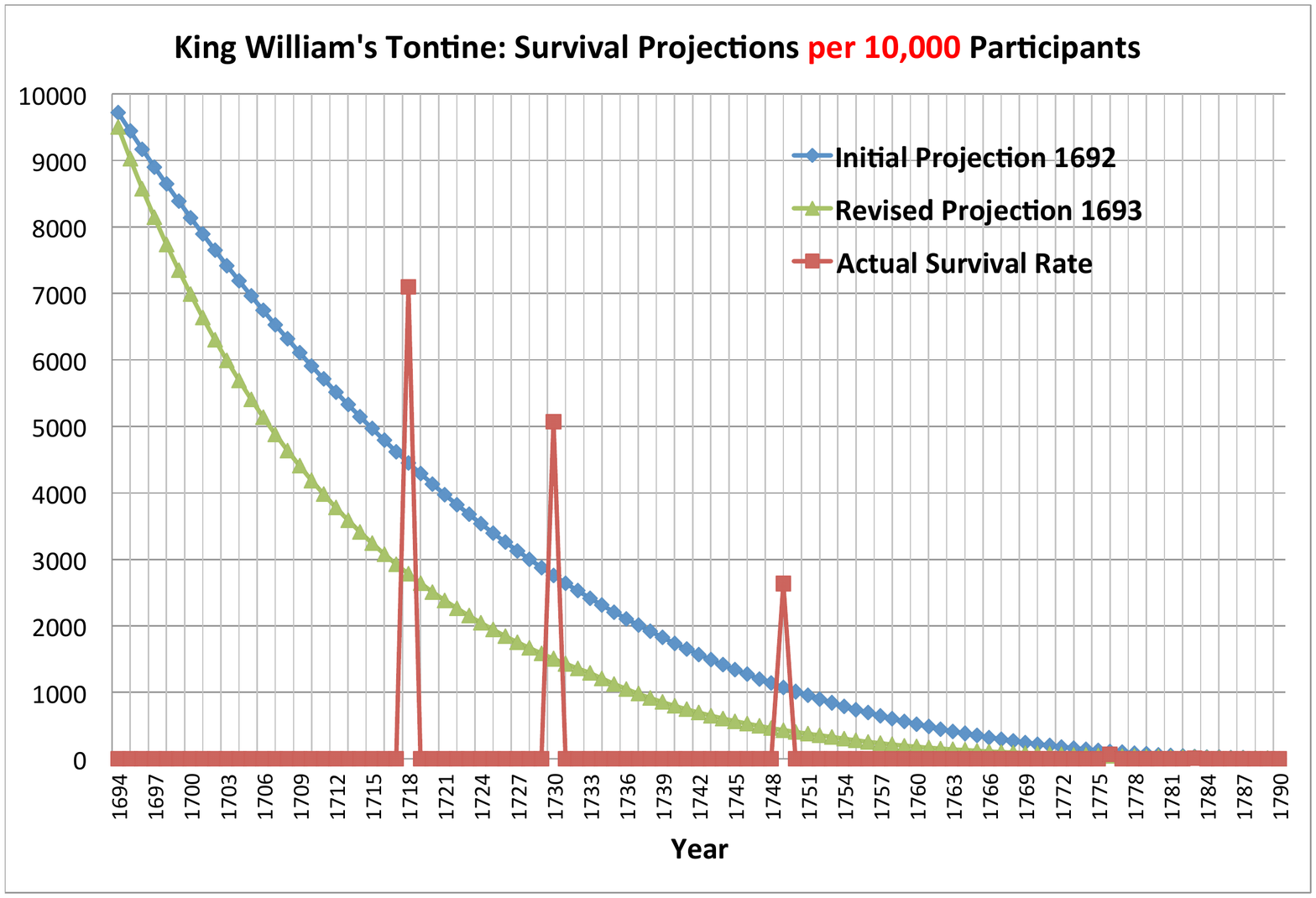} 
\caption{The initial projections made by the promoters of the tontine in late 1692 were too optimistic. The revised projections made in mid-1693 were wildly optimistic. In fact, the 70\% survival rate observed by the year 1718 led to accusations of fraud. No other data points are available. Data source: Lewin (2003), Walford (1871). Compiled by The IFID Centre.}
\label{fig1}
\end{center}
\end{figure}

\begin{figure}[here]
\begin{center}
\includegraphics[width=1.0\textwidth]{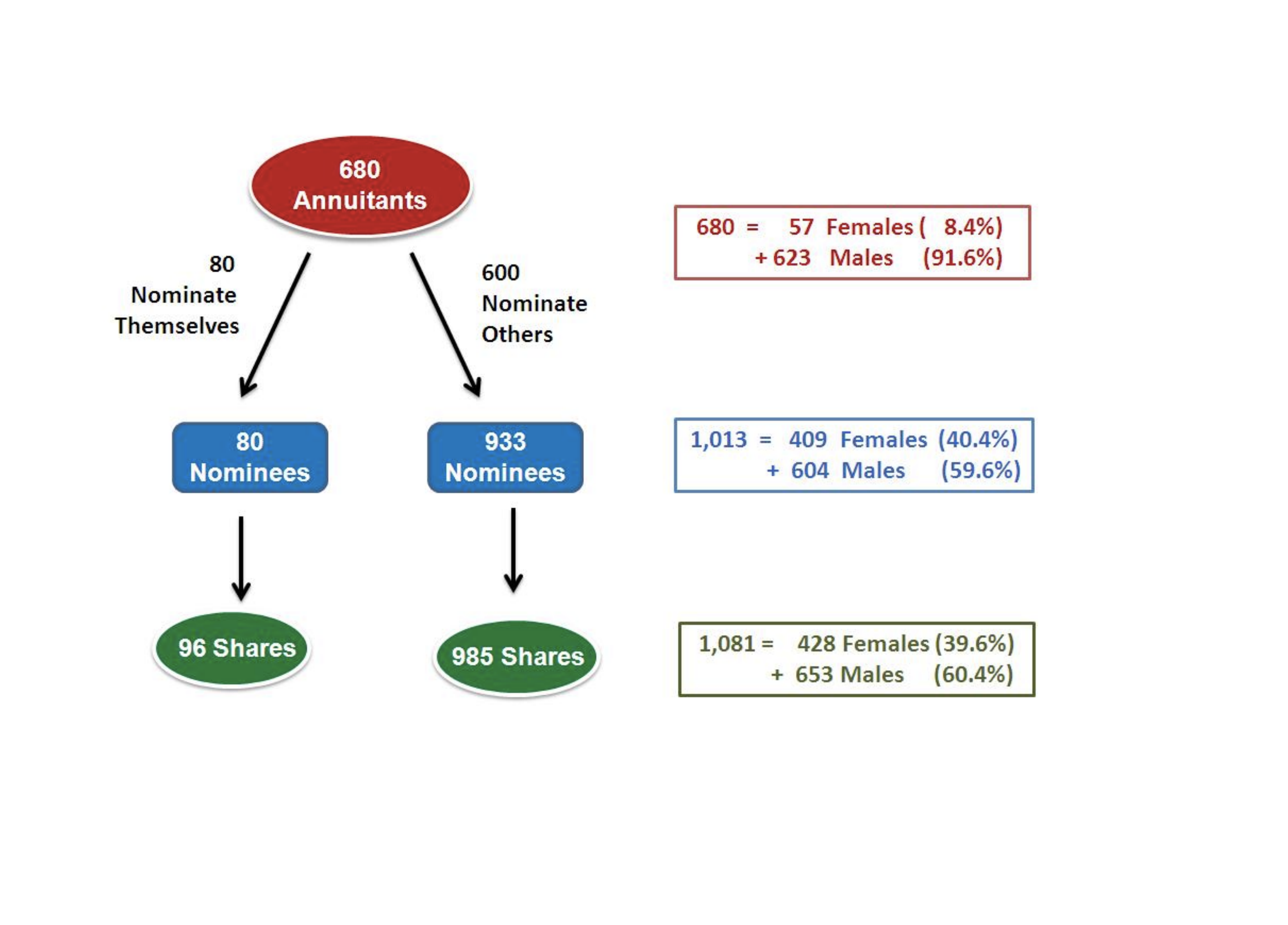} 
\caption{Final verified tally of the annuitants (=investors), nominees and number of shares purchased in King William's tontine of 1693, as well as their gender composition. Compiled by The IFID Centre based on the original records printed in Howard (1694).}
\label{fig1b}
\end{center}
\end{figure}

\begin{figure}[here]
\begin{center}
\includegraphics[width=1.0\textwidth]{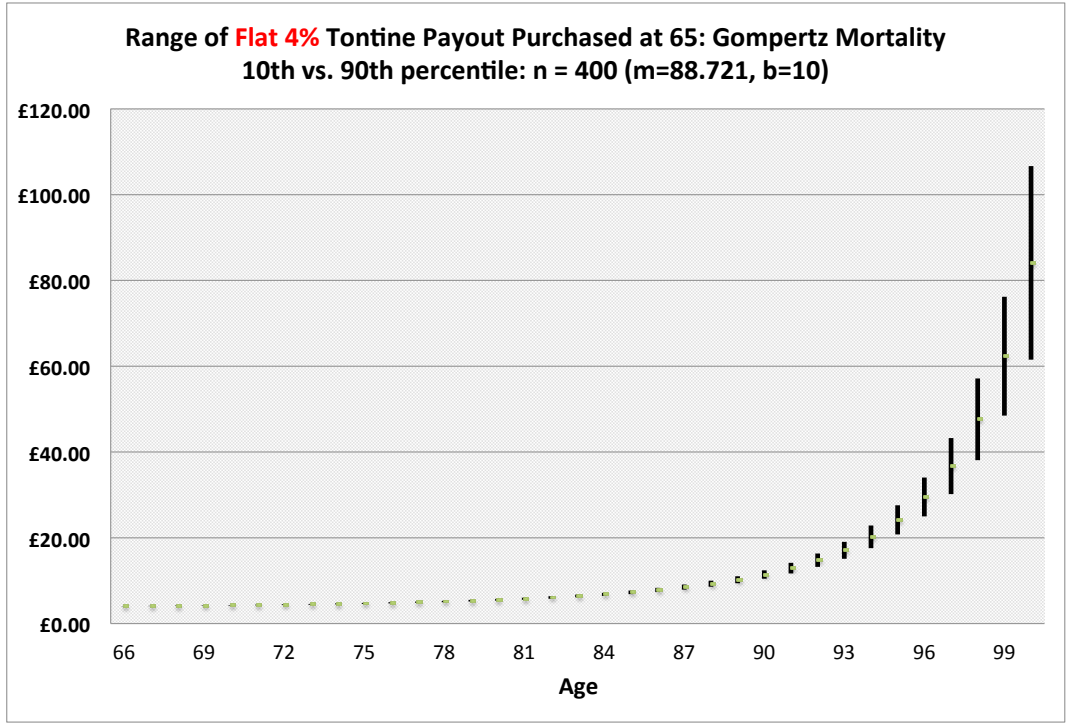} 
\caption{The range of $d_0=4\%$ tontine dividends during the first few decades or retirement is relatively low and predictable for a pool size in the hundreds. The dividends increase exponentially at later ages and the 80\% range is much wider as well. But this is not the only way to construct a tontine.}
\label{fig2}
\end{center}
\end{figure}

\begin{figure}[here]
\begin{center}
\includegraphics[width=1.0\textwidth]{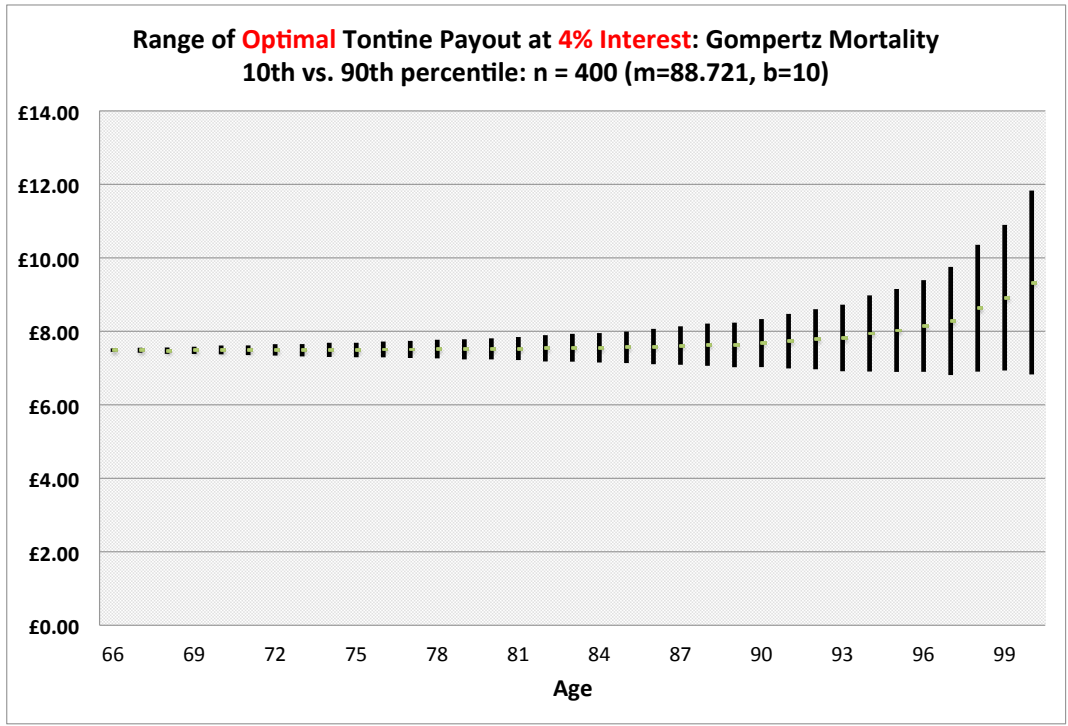} 
\caption{The optimal tontine pays survivors a cash-value that is expected to remain relatively constant over time, conditional on survival. Although the 80\% range of outcomes does increase at higher ages given the inherent uncertainty in the number of survivors from an initial pool of $n=400$. This structure is optimal for logarithmic $\gamma=1$ utility and nearly optimal for all other levels of Longevity Risk Aversion.}
\label{fig3}
\end{center}
\end{figure}

\begin{figure}[here]
\begin{center}
\includegraphics[width=1.0\textwidth]{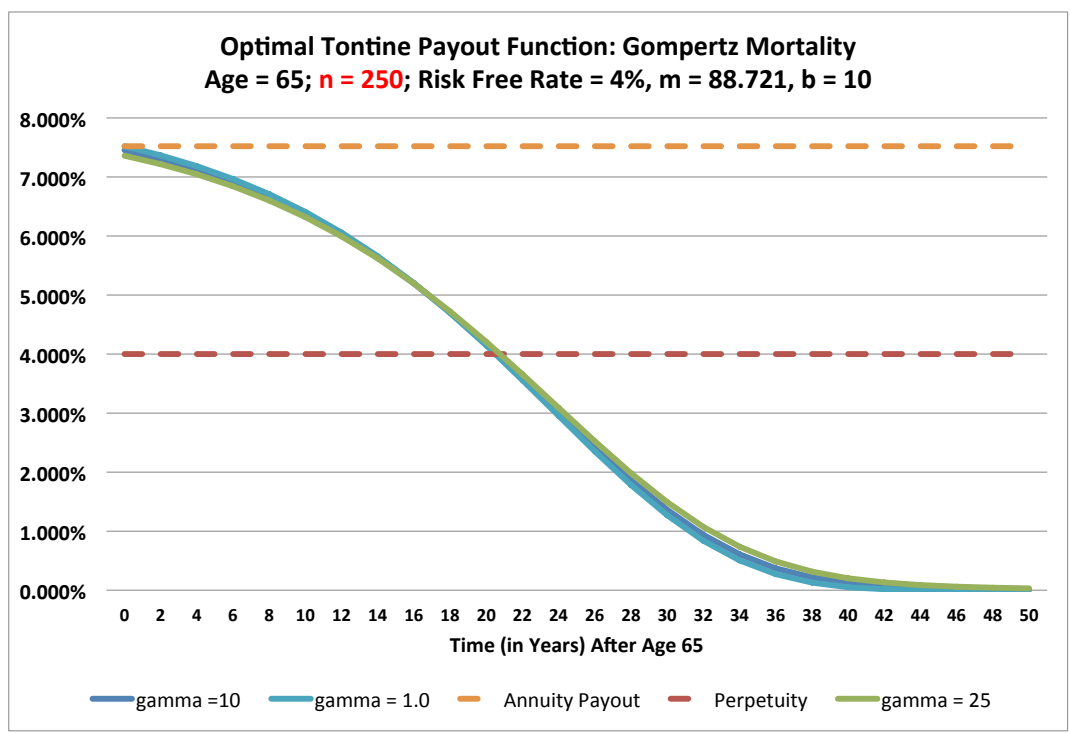} 
\caption{The difference between the optimal tontine payout function for different levels of Longevity Risk Aversion $\gamma$ is barely noticeable when the tontine pool size is greater than $N=250$, mainly due to the effect of the law of large numbers. This curve traces the \emph{minimum} dividend that a survivor can expect to receive at various ages. The expected is (obviously) much higher.}
\label{fig4}
\end{center}
\end{figure}

\begin{figure}[here]
\begin{center}
\includegraphics[width=1.0\textwidth]{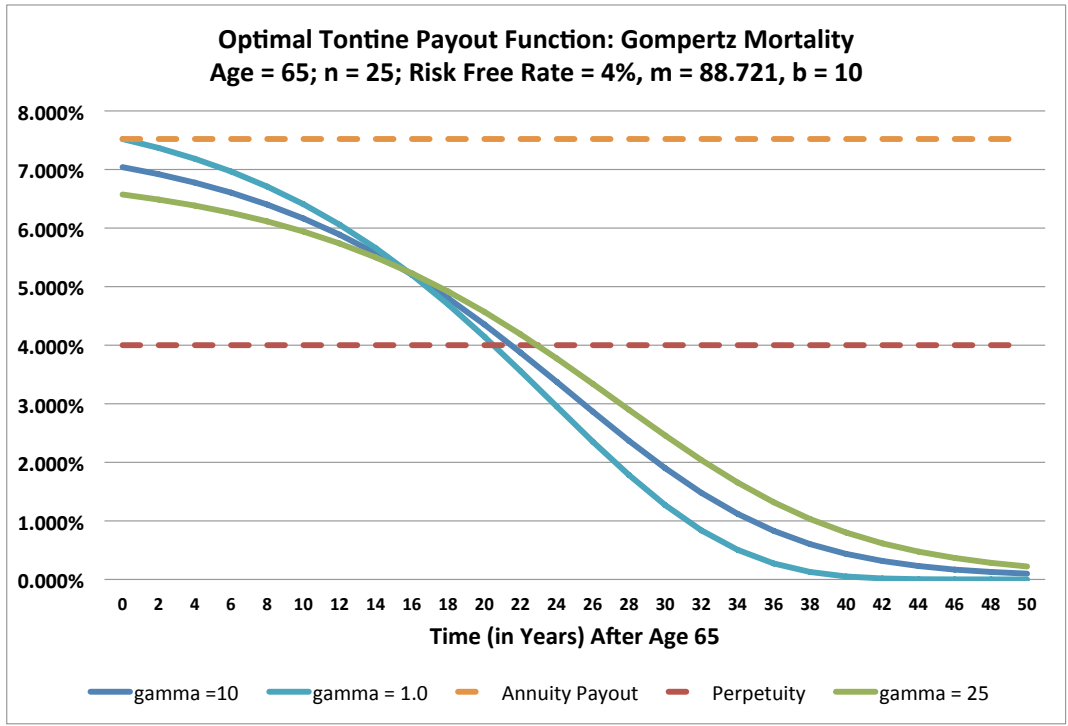} 
\caption{The optimal tontine payout function for someone with logarithmic $\gamma=1$ utility starts-of paying the exact same rate as a life annuity regardless of the number of participants in the tontine pool. But, for higher levels of longevity risk aversion $\gamma$ and a relatively smaller tontine pool, the function starts-off at a lower value and declines at a slower rate}
\label{fig5}
\end{center}
\end{figure}

\begin{figure}[here]
\begin{center}
\includegraphics[width=1.0\textwidth]{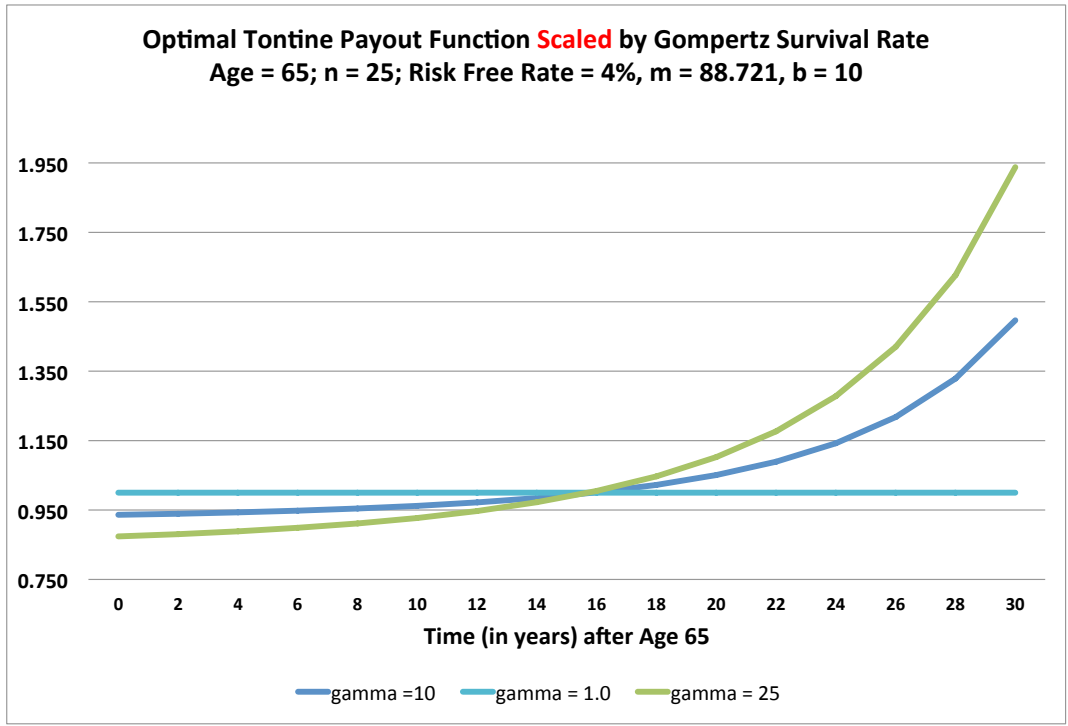} 
\caption{For a relatively smaller tontine pool size, the retiree who is highly averse to longevity risk $\gamma=25$ will want a guaranteed minimum payout rate (GMPR) at advanced ages that is higher than his or her projected survival probability. In exchange they will accept lower GMPR at lower ages. In contrast, the logarithmic utility maximizer will select a GMPR that is exactly equal to the projected survival probability.}
\label{fig6}
\end{center}
\end{figure}

\newpage

\begin{figure}[here]
\begin{center}
\includegraphics[width=1.0\textwidth]{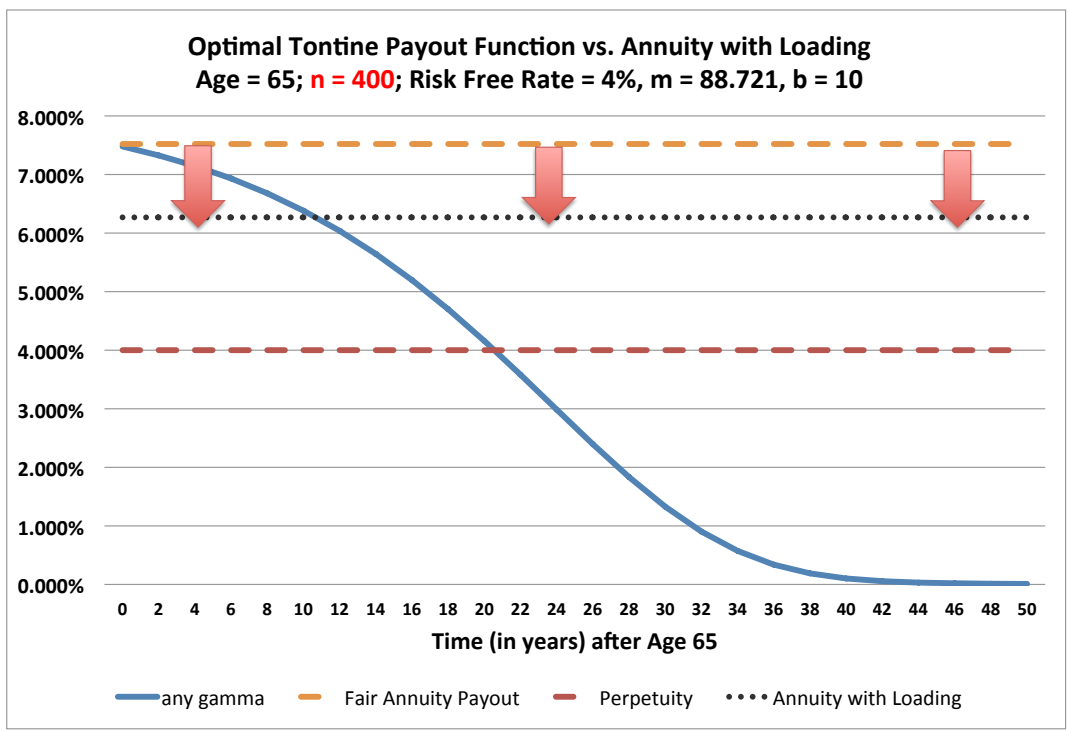} 
\caption{An actuarially fair life annuity that guarantees 7.5\% for life starting at age 65 provides more utility than an optimal tontine regardless of Longevity Risk Aversion (LoRA) or the size of the tontine pool. But, once an insurance loading is included, driving the yield under the initial payout from the optimal tontine, the utility of the life annuity might be lower. The indifference loading is $\delta$ and reported in table \ref{table07} }
\label{fig7}
\end{center}
\end{figure}

\end{document}